\documentclass[12pt, reqno]{amsart}

\usepackage{amsmath, amssymb}
\usepackage{stmaryrd}
\usepackage{simplewick}

\newcommand{\wh}[1]{\widehat{#1}}
\newcommand{\ii}{\mathrm{i}}
\newcommand{\wt}{\widetilde}

\newtheorem{definition}{Definition}
\newtheorem{remark}{Remark}
\newtheorem{proposition}{Proposition}
\newtheorem{theorem}{Theorem}

\begin{document}

\title[Noncommutative Planck Constants]{Distortion of the Poisson Bracket by the Noncommutative Planck Constants}

\author[A. E. Ruuge, F. Van Oystaeyen]{Artur E. Ruuge \and Freddy Van Oystaeyen}

\address{
Department of Mathematics and Computer Science, 
University of Antwerp, 
Middelheim Campus Building G, 
Middelheimlaan 1, B-2020, 
Antwerp, Belgium}

\email{artur.ruuge@ua.ac.be; fred.vanoystaeyen@ua.ac.be}


\maketitle

\begin{abstract} 
In this paper 
we introduce a kind of ``noncommutative neighbourhood'' 
of a semiclassical parameter corresponding to the Planck constant. 
This construction is defined as a certain filtered and graded algebra 
with an infinite number of generators indexed by planar binary leaf-labelled trees. 
The associated graded algebra (the classical shadow) is interpreted as 
a ``distortion'' of the algebra of classical observables of a physical system. 
It is proven that there exists a $q$-analogue of the Weyl quantization, 
where $q$ is a matrix of formal variables, which induces a nontrivial 
noncommutative analogue of a Poisson bracket on the classical shadow. 
\end{abstract}

\section{Introduction}

In the present paper we describe a mathematical construction 
which can be perceived as a kind of \emph{noncommutative neighbourhood} 
of the parameter $\hbar \to 0$ of the semiclassical approximation of quantum theory 
(the Planck ``constant''). 
This construction can be of interest in noncommutative algebraic geometry, as well as in mathematical 
physics, and, informally speaking, it is linked to an idea of a \emph{``quantization on a noncommutative space''}. 

In quantum mechanics, if we speak about quantization, then this normally implies that we have a 
linear map $Q: \mathcal{A} \to \mathcal{B}$ between two algebras defined 
over a ring of formal power series $\mathbb{C} [[\hbar]]$, where $\mathcal{A}$ is commutative 
(the classical observables), and $\mathcal{B}$ is noncommutative (the quantum observables). 
Since the map $Q$ does not need to be an algebra homomorphism, the multiplication on $\mathcal{B}$ 
induces a noncommutative associative product $*_{\hbar}$ on $\mathcal{A}$, 
\begin{equation*} 
f *_{\hbar} g = f g + \hbar B_1 (f, g) + \hbar^2 B_2 (f, g) + \dots, 
\end{equation*}
for $f, g \in \mathcal{A}$, where $B_1, B_2, \dots$ are bilinear maps. 
The first map $B_1$ gives rise to a Poisson bracket on $\mathcal{A}$, 
\begin{equation*} 
\lbrace f, g \rbrace := B_1 (f, g) - B_1 (g, f), 
\end{equation*}
and, in fact, ``quantization'' is always a quantization in the direction of a given 
Poisson structure $\lbrace -, - \rbrace_{\mathrm{cl}}$ on $\mathcal{A}$ (i.e. one is asked to construct a map $Q$ 
such that $\lbrace -, - \rbrace = \lbrace -, - \rbrace_{\mathrm{cl}}$).  

It is natural to consider a more general case where both algebras $\mathcal{A}$ and $\mathcal{B}$ can be noncommutative. 
Then, in a certain sense, a quantization $Q: \mathcal{A} \to \mathcal{B}$ consists in making things ``more noncommutative''. 
Naively, this might look as a technical generalization of what already exists, but 
in reality one runs into a serious conceptual problem here. 
\emph{How do we define a noncommutative analogue of the Poisson bracket?} 
Given an arbitrary noncommutative algebra $\mathcal{A}$, 
if we impose the antisymmetry, the Leibniz rule, and the Jacobi identity (i.e. the usual axioms 
of the Poisson bracket), then it can easily turn out that there are no ``interesting'' Poisson structures on it 
\cite{Kubo,Xu,Yao_Ye_Zhang}, 
and essentially the only Poisson structure is given by the commutator $[f, g] := f g - g f$, $f, g \in \mathcal{A}$.  
Therefore, there exist different approaches to define a noncommutative generalization of the Poisson bracket. 

One of the possibilities is to modify the axioms by introducing \emph{twists} (twisted antisymmetry, 
twisted Leibniz rule, twisted Jacobi identity), what is a common practice, for instance, in the theory of 
coloured Lie algebras and their representations 
\cite{Hartwig_Larsson_Silvestrov,Richard_Silvestrov}. 
Another possibility is to perceive the problem in terms 
of homotopy theory and to generalize the Jacobi identity ``up to homotopy'' 
\cite{Alekseev_KosmannSwarzbach,Alekseev_KosmannSwarzbach_Meinrenken,Farkas,Fialowski_Penkava,Penkava}.  
The Poisson bracket then gets replaced with an infinite collection of operations 
\begin{equation*} 
\lbrace -, -, \dots, - \rbrace_{n} : \mathcal{A}^{\otimes n} \to \mathcal{A}, 
\end{equation*}
of different arity $n = 1, 2, \dots$, satisfying the $L_{\infty}$-algebra kind of axioms. 
The third possibility, which has recently received some additional attention in the literature
\cite{VandenBergh,Bocklandt_LeBruyn,CrawleyBoevey_Etingof_Ginzburg,Ginzburg,Pichereau_VandeWeyer}, 
is to consider the ``double Poisson'' bracket, 
\begin{equation*} 
\langle -, - \rangle : \mathcal{A}^{\otimes 2} \to \mathcal{A}^{\otimes 2}, 
\end{equation*}
i.e. a bracket with the values in $\mathcal{A} \otimes \mathcal{A}$, rather than in $\mathcal{A}$. 
The latter seems to be a reasonable approach for the algebras which are ``far from commutative''. 
For example, 
there exists a canonical double Poisson structure on 
the free noncommutative algebra $\mathbb{C} \langle \xi_1, \xi_2, \dots, \xi_n \rangle$ on 
a finite number of symbols $\xi_1, \xi_2, \dots, \xi_n$. 
To keep the story short, the idea of a noncommutative Poisson bracket leaves some space for creativity. 

The main construction of the present paper can be perceived as follows. 
If we look at an abstract $d$-dimensional quantum mechanical system with coordinates 
$\wh{x}_1, \wh{x}_2, \dots, \wh{x}_d$ and momenta $\wh{p}_1, \wh{p}_2, \dots, \wh{p}_d$, 
then we have the canonical commutation relations 
\begin{equation*} 
[\wh{p}_i, \wh{x}_{j}] = - \ii \hbar, \quad 
[\wh{p}_i, \wh{p}_j] = 0 = [\wh{x}_i, \wh{x}_j], 
\end{equation*}
for $i, j \in [d] := \lbrace 1, 2, \dots, d \rbrace$, where $\hbar$ is central. 
\emph{Why do we put the commutator $[\wh{p}_i, \wh{x}_j]$ into the centre?} 
If one considers the bracketed expressions over the generators $\wh{x}_i$, $\wh{p}_j$, $i, j \in [d]$, 
which have a length of at least three, then this yields just $0, 0, \dots$. 
Essentially, we suggest to go to another extreme: let us \emph{declare every next commutator a new variable}. 
In other words, let us replace $\hbar$ with an infinite collection of symbols 
$\lbrace \wh{\hbar}_{\Gamma} \rbrace_{\Gamma}$, where $\Gamma$ varies over all finite planar binary trees 
(i.e. all possible bracketings) equipped with a leaf labelling from the set $[2 d]$, 
\begin{equation*} 
\hbar \, \longrightarrow \, 
\wh{\hbar}_i, \quad 
\wh{\hbar}_{[i, j]}, \quad 
\wh{\hbar}_{[[i, j], k]}, \quad
\wh{\hbar}_{[[[i, j], k], l]}, \quad
\wh{\hbar}_{[[i, j], [k, l]]}, \quad 
\dots, 
\end{equation*}
where $i, j, k, l, \dots$ vary over $[2 d] := \lbrace 1, 2, \dots, 2 d \rbrace$. 
The commutation relations are as follows: 
\begin{equation*} 
[\wh{\hbar}_i, \wh{\hbar}_j] = \wh{\hbar}_{[i, j]}, \quad 
[\wh{\hbar}_{[i, j]}, \wh{\hbar}_{k}] = \wh{\hbar}_{[[i, j], k]}, \quad 
[\wh{\hbar}_{[i, j]}, \wh{\hbar}_{[k, l]}] = \wh{\hbar}_{[[i, j], [k, l]]}, \quad 
\dots, 
\end{equation*} 
for $i, j, k, l, \dots \in [2 d]$. 
This yields an infinite dimensional noncommutative algebra, which we denote 
$\mathcal{E}_{2 d} (\hbar)$ and perceive as a ``noncommutative neighbourhood'' of $\hbar$, 
having in mind an analogy with the terminology of 
\cite{Kapranov} 
for the noncommutative manifolds. 
One can now identify the generators corresponding to the coordinates and momenta with the 
\emph{noncommutative Planck constants} $\wh{\hbar}_{i}$, $i \in [2 d]$, corresponding to the trees with only one leaf: 
\begin{equation*} 
\begin{aligned}
\wh{p}_1 &= \wh{\hbar}_1,& 
\wh{p}_2 &= \wh{\hbar}_2,& 
\dots \quad 
\wh{p}_{d} &= \wh{\hbar}_{d}, \\
\wh{x}_1 &= \wh{\hbar}_{d + 1},&   
\wh{x}_2 &= \wh{\hbar}_{d + 2},& 
\dots \quad 
\wh{x}_{d} &= \wh{\hbar}_{2 d}. 
\end{aligned}
\end{equation*}

In the present paper, we construct a natural $q$-deformation $\mathcal{E}_{2 d}^{q} (\hbar)$ of the 
algebra $\mathcal{E}_{2 d} (\hbar)$ corresponding to a $2 d \times 2 d$ matrix of formal variables, $q = \| q_{i, j} \|$, 
$q_{i, i} = 1$, $q_{i, j} = q_{j, i}^{- 1}$, where $i, j \in [2 d]$. 
In particular, every commutator $[\wh{\hbar}_i, \wh{\hbar}_j]$, $i, j \in [2 d]$, 
in the defining relations of $\mathcal{E}_{2 d} (\hbar)$ 
gets replaced with a $q$-commutator $[\wh{\hbar}_i, \wh{\hbar}_j]_{q} := \wh{\hbar}_i \wh{\hbar}_j - q_{j, i} \wh{\hbar}_j \wh{\hbar}_i$. 
After that we consider a series of truncations 
$(\mathcal{E}_{2 d}^{q} (\hbar))_{\leqslant N}$, $N = 1, 2, 3, \dots$, 
obtained by factoring out of the ideals in  
$\mathcal{E}_{2 d}^{q} (\hbar)$, 
generated by the symbols $\wh{\hbar}_{\Gamma}$ with the number of leaves 
$|\Gamma| > N$. 
Note, that the commutation relations of $(\mathcal{E}_{2 d}^{\mathbf{1}})_{\leqslant 2}$, where 
$\mathbf{1}$ is a $2 d \times 2 d$ matrix with all entries equal to one, appears in 
\cite{Castro}, but in a different context. 
We prove, that for every $N = 1, 2, 3, \dots$, there exists an analogue 
of the Weyl quantization map (the $q$-Weyl quantization), which we define as a linear map 
\begin{equation*} 
W_{N}^{q} : \mathrm{gr} \big( (\mathcal{E}_{2 d}^{q} (\hbar))_{\leqslant N} \big) \to 
(\mathcal{E}_{2 d}^{q} (\hbar))_{\leqslant N}, 
\end{equation*}
where $\mathrm{gr} ( - )$ denotes taking the associated graded with respect to a filtration 
induced by the degrees of the generators 
\begin{equation*} 
\mathrm{deg} (\wh{\hbar}_{\Gamma}) := |\Gamma| - 1. 
\end{equation*}
It turns out, that this map $W_{N}^{q}$ has an nontrivial property 
in comparison to the well known Weyl quantization map in quantum mechanics. For an abstract physical system 
with two degrees of freedom ($d = 2$), the latter is a linear map 
\begin{equation*} 
Q_{\mathrm{Weyl}}: \mathbb{C} [p, x] [\hbar] \to 
\mathbb{C} \langle \wh{p}, \wh{x} \rangle [\hbar]/ (\wh{p} \wh{x} - \wh{x} \wh{p} = - \ii \hbar), 
\end{equation*} 
where one puts an $\hbar$-adic filtration on the domain of $Q_{\mathrm{Weyl}}$. 
It is linear not just over $\mathbb{C}$, but over the ring $\mathbb{C}[\hbar]$, 
so, in particular, we have: $Q_{\mathrm{Weyl}} (\hbar f) = \hbar Q_{\mathrm{Weyl}} (f)$, for any monomial $f$ in $x$ and $p$.  
On the other hand, 
if we denote $\hbar_{\Gamma}$ the canonical image of $\wh{\hbar}_{\Gamma}$ in the associated graded, and 
take, for example, a generic monomial $u$ in the variables $\hbar_k$, $k \in [2 d]$, 
and a symbol $\hbar_{[i, j]}$, where $i, j \in [2 d]$, $i \not = j$, 
then we obtain: 
$W_{N}^{q} (\hbar_{[i, j]} u) \not = \wh{\hbar}_{[i, j]} W_{N}^{q} (u)$, if $N > 2$. 
In this sense, the generators $\wh{\hbar}_{\Gamma}$ corresponding to the trees $\Gamma$ with the number of leaves $|\Gamma| \geqslant 2$ 
are just as good on the role of ``observables'' as the generators $\wh{\hbar}_{i}$, $i \in [2 d]$. 
We suggest to refer to this fact as to \emph{distortion of quantization} by the noncommutative Planck constants. 
If one considers a representation theory of $(\mathcal{E}_{2 d}^{q} (\hbar))_{\leqslant N}$, say, for $N = 3$, 
then one can define an ``observable'' which \emph{does not have an analogue} in quantum theory. 
Does this only distort the conventional description of our physical reality, or does it make some physical sense? 

For the Weyl quantization map $Q_{\mathrm{Weyl}}$, one can extract the Poisson bracket on 
$\mathbb{C} [x, p] [\hbar]$ as a first semiclassical correction to the multiplication of classical observables: 
\begin{equation*} 
Q_{\mathrm{Weyl}} \Big( f g - \frac{\ii \hbar}{2} \lbrace f, g \rbrace + O (\hbar^{2 + |f| + |g|}) \Big) = 
Q_{\mathrm{Weyl}} (f) Q_{\mathrm{Weyl}} (g), 
\end{equation*}
where $O (\hbar^n)$ stands for an element of degree $n$ in the $\hbar$-adic filtration on 
$\mathbb{C} [p, x] [\hbar]$, $\mathrm{deg} (\hbar) = 1$, 
$\mathrm{deg} (p) = 0$, $\mathrm{deg} (x) = 0$, 
and $|f|$ and $|g|$ denote the degrees of homogeneous elements 
$f, g \in \mathbb{C} [p, x] [\hbar]$. 
This bracket $\lbrace -, - \rbrace$ is bilinear not just over $\mathbb{C}$, but over $\mathbb{C} [\hbar]$. 
In the present paper we perform a similar extraction of a bracket $\langle -, - \rangle_{N}^{q}$
from the $q$-Weyl quantization map $W_{N}^{q}$, 
for $N = 1, 2, 3, \dots$, 
\begin{equation*} 
W_{N}^{q} \big( u v + \langle u, v \rangle_{N}^{q} + O (2 + |u| + |v|) \big) = W_{N}^{q} (u) W_{N}^{q} (v), 
\end{equation*}
where $u, v \in \mathrm{gr} \big( (\mathcal{E}_{2 d}^{q} (\hbar))_{\leqslant N} \big)$ 
are homogeneous elements of degrees $|u|$ and $|v|$, respectively, and 
$O (n)$, $n = 0, 1, 2, \dots$, stands for an element of degree at least $n$ in the filtration induced by the grading. 
Note, that this corresponds to considering ``semiclassics'' with $\hbar \to 0$ and $q$ fixed, 
which is different from \cite{Beggs_Majid,Chaichian_Demichev_Kulish}. 
The distortion of quantization has its counterpart on the properties of the corresponding 
bracket $\langle -, - \rangle_{N}^{q}$: it is linear only over $\mathbb{C}$. 
We call this fact a \emph{distortion of the Poisson bracket} by the noncommutative Planck constants 
and describe its properties as of a natural candidate for a noncommutative Poisson bracket 
(the ``direction'' of quantization on a noncommutative geometry).

\section{Planar binary trees} 

It is convenient to realize the collection of all planar binary trees as follows. 
Take your favourite singleton $\lbrace * \rbrace$ (i.e. $*$ is just a symbol). 
Set $\mathcal{Y}_{0} := \emptyset$ and 
$\mathcal{Y}_{1} := \lbrace * \rbrace$. 
For every $n = 2, 3, 4, \dots$, define $\mathcal{Y}_{n}$ recursively as a set of all pairs 
$T = (u, v)$, where $u \in \mathcal{Y}_{p}$, $v \in \mathcal{Y}_{q}$, $p, q \geqslant 1$, $p + q = n$. 
The tree $L (T) := u$ is termed the \emph{left branch} of $T$, the tree 
$R (T) := v$ is termed the \emph{right branch} of $T$, 
and $|T| := n$ is termed the \emph{number of leaves} in $T$. 
Denote $\mathcal{Y} := \bigsqcup_{n = 0}^{\infty} \mathcal{Y}_{n}$.

The planar binary trees $\mathcal{Y}$ naturally encode all possible bracketings over a set of symbols. 
Suppose we have an alphabet $\Omega$. 
Consider a word of length three: $w = x_1 x_2 x_3$, where $x_1, x_2, x_3 \in \Omega$. 
The set $\mathcal{Y}_{3}$ contains two elements, $((*, *), *)$ and $(*, (*, *))$. 
The first element corresponds to $((x_1 x_2) x_3)$, and the second element corresponds to $(x_1 (x_2 x_3))$. 
The set $\mathcal{Y}_4$ contains already five elements, and if we take, for example, $(*, ((*, *), *)) \in \mathcal{Y}_{4}$ 
and a word $w' = x_1 x_2 x_3 x_4$ over $\Omega$, then the corresponding bracketed expression is 
$(x_1 ((x_2 x_3) x_4))$, etc. 
The number of elements $\# \mathcal{Y}_{n} = C_{n - 1}$, for $n = 1, 2, 3, \dots$, 
is given by the \emph{Catalan numbers}, 
\begin{equation*} 
C_{m} := \frac{(2 m)!}{(m + 1)! m!}, 
\end{equation*}   
where $m = 0, 1, 2, \dots$. 

Fix an alphabet $\Omega$. A leaf-labelled planar binary tree can be perceived as a pair 
$\Gamma = (T, w)$, where $T \in \mathcal{Y}$, and $w$ is a word of length $|T|$ over $\Omega$. 
For example, the set $\mathcal{Y}_2$ contains only one element $(*, *)$, and if we take $x_1, x_2 \in \Omega$,  
then $\Gamma' = ((*, *), x_1 x_2)$ is a leaf-labelled tree, 
such that the leaf on the left branch is labelled with $x_1$, and 
the leaf on the right branch is labelled with $x_2$. 
Denote $\mathcal{Y} (\Omega)$ the set of all leaf labelled planar binary trees over the alphabet $\Omega$. 
The \emph{left branch} of a leaf-labelled tree 
$\Gamma = (T, w) \in \mathcal{Y} (\Omega)$, $w = x_1 x_2 \dots x_n$, $x_i \in \Omega$, $i \in [n]$, 
is defined as $L (\Gamma) := (L (T), w')$, where $w' = w_1 w_2 \dots w_m$, $m = |L (T)|$, and 
the \emph{right branch} is defined as $R (\Gamma) := (R (T), w'')$, where $w'' = w_{m + 1} w_{m + 2} \dots w_n$.

The set $\mathcal{Y}$ of all planar binary trees can be totally ordered as follows. 
Let $T, T' \in \mathcal{Y}$, $T \not = T'$. 
If $|T| < |T'|$, then set $T \prec T'$, and 
if $|T| > |T'|$, then set $T \succ T'$. 
In case $|T| = |T'|$, 
both trees must have at least two leaves, since $\# \mathcal{Y}_1 = 1$, while $T \not = T'$.  
Therefore, one can 
look at the left and right branches: 
$T = (L (T), R (T))$ and $T' = (L (T'), R (T'))$.  
Set $T \prec T'$, if $L (T) \prec L (T')$, and 
$T \succ T'$, if $L (T) \succ L (T')$. 
If $L (T) = L (T')$, then compare the right branches and set 
$T \prec T'$ if and only if $R (T) \prec R (T')$. 
Since the number of leaves in the branches is always strictly smaller than the 
number of leaves in the whole tree, this defines recursively a total order $\prec$ on $\mathcal{Y}$.  

Assume that there is a total order $\prec$ on the alphabet $\Omega$ 
(here we overload the notation $\prec$ and denote the total orders on different sets with the same 
symbol, assuming the apparent ambiguity is always resolved by the context). 
It induces a total order on the words: 
$x_1 x_2 \dots x_m \prec y_1 y_2 \dots y_n$ if and only if 
$(x_1 \prec y_1) \vee (x_1 = y_1 \,\&\, x_2 \dots x_m \prec y_2 \dots y_n)$, 
where $x_i, y_j \in \Omega$, $i \in [m]$, $j \in [n]$.   
Therefore the leaf-labelled trees $\mathcal{Y} (\Omega)$ can be totally ordered as well: 
\begin{equation*} 
(T, w) \prec (T', w') \quad :\Leftrightarrow \quad 
(T \prec T') \vee (T = T' \,\&\, w \prec w'), 
\end{equation*}
where $(T, w), (T', w') \in \mathcal{Y} (\Omega)$. 
If the total number of letters is finite, i.e. $\# \Omega < \infty$, then 
there is a unique bijection $\mathcal{Y} (\Omega) \simeq \mathbb{Z}_{> 0}$, 
which respects the order,
\begin{equation*} 
\nu : \mathcal{Y} (\Omega) \overset{\sim}{\to} \mathbb{Z}_{> 0}, \quad 
\Gamma \prec \Gamma' \,\Leftrightarrow\, \nu (\Gamma) < \nu (\Gamma'), 
\end{equation*}
where $\Gamma, \Gamma' \in \mathcal{Y} (\Omega)$. 
We will use this map later to define the standard ordering for the monomials in 
the noncommutative Planck constants $\wh{\hbar}_{\Gamma}$, $\Gamma \in \mathcal{Y} ([2 d])$. 
The total order on $[2 d]$ is assumed to be $<$.  

Note, that there is a naturally defined operation on $\mathcal{Y}$ corresponding to the \emph{concatenation} of trees, 
$T \vee T' := T''$, where $T'' = (T, T') \in \mathcal{Y}$. 
Since $|T \vee T'| = |T| + |T'|$, one has also a concatenation on the leaf-labelled trees: 
\begin{equation*} 
(T, w) \vee (T', w') := (T \vee T', w w'), 
\end{equation*} 
for every $(T, w), (T', w') \in \mathcal{Y} (\Omega)$. 
The latter generalizes naturally to a \emph{composition} of trees as follows. 
If one takes a collection of leaf-labelled trees $\Gamma_i = (T_i, w_i) \in \mathcal{Y} (\Omega)$, 
$i \in [n]$, and a leaf-labelled tree $\Gamma = (T, \bar \sigma) \in \mathcal{Y} ([n])$, 
where $\bar \sigma := (\sigma (1), \sigma (2), \dots, \sigma (n))$ is a word over $[n]$ 
corresponding to a permutation $\sigma \in S_n$, then deleting, for every $i \in [n]$, 
the $i$-th leave in $\Gamma$ and inserting in its place   
the tree $\Gamma_i$, one obtains a tree $T'$ with a labelling 
$w' = w_1 w_2 \dots w_n$. Denote the result $\Gamma' \equiv (T', w')$ as $T_{\sigma} (\Gamma_1, \Gamma_2, \dots, \Gamma_n)$. 
If $\sigma = \mathit{id}$, then write just $\Gamma' = T (\Gamma_1, \Gamma_2, \dots, \Gamma_n)$. 
In particular, $\Gamma_1 \vee \Gamma_2 = (*, *) (\Gamma_1, \Gamma_2)$.

\section{The bracketing algebra}

Let $q = \| q_{i, j} \|$ a $2 d \times 2 d$ matrix of formal variables $q_{i, j}$
(where $d$ is a fixed positive integer) satisfying $q_{i, i} = 1$ and $q_{i, j} = q_{j, i}^{-1}$, $i, j \in [2 d]$. 
Extend this notation as follows: 
\begin{equation*} 
q_{\Gamma, \Gamma'} := \prod_{\alpha = 1}^{n} \prod_{\beta = 1}^{m} q_{i_{\alpha}, j_{\beta}},  
\end{equation*}
for every $\Gamma, \Gamma' \in \mathcal{Y} ([2 d])$, 
where $(i_1, i_2, \dots, i_n)$ is the leaf-labelling of $\Gamma$, $n = |\Gamma|$ 
(we write the words over $[2 d]$ as sequences separated by commas), and 
$(j_1, j_2, \dots, j_m)$ is the leaf-labelling of $\Gamma'$, $m = |\Gamma'|$. 
Observe, that 
\begin{equation*} 
q_{\Gamma, \Gamma} = 1, \quad 
q_{\Gamma, \Gamma'} = q_{\Gamma', \Gamma}^{-1}, 
\end{equation*}
for $\Gamma, \Gamma' \in \mathcal{Y} ([2 d])$.

Denote $K_{2 d}^{(q)}$ the ring of polynomials with complex coefficients in the variables  
$\lbrace q_{i, j} \rbrace_{i < j}$ and $\lbrace q_{i, j}^{- 1} \rbrace_{i < j}$, $i, j \in [2 d]$. 
Consider an algebra $\mathcal{E}_{2 d}^{q} (\hbar)$ over $K_{2 d}^{(q)}$
generated by an infinite collection of symbols 
$\lbrace \wh{\hbar}_{\Gamma} \rbrace_{\Gamma \in \mathcal{Y} ([2 d])}$ satisfying the relations 
\begin{equation*} 
[\wh{\hbar}_{\Gamma}, \wh{\hbar}_{\Gamma'}]_{q} := 
\wh{\hbar}_{\Gamma} \wh{\hbar}_{\Gamma'} - q_{\Gamma', \Gamma} 
\wh{\hbar}_{\Gamma'} \wh{\hbar}_{\Gamma} = \wh{\hbar}_{\Gamma \vee \Gamma'}, 
\end{equation*}
where $\Gamma, \Gamma' \in \mathcal{Y} ([2 d])$. 
Observe, that the notation implies 
\begin{equation} 
\label{eq:hbar_convention}
\wh{\hbar}_{\Gamma \vee \Gamma} = 0, \quad 
\wh{\hbar}_{\Gamma \vee \Gamma'} = -q_{\Gamma', \Gamma} \wh{\hbar}_{\Gamma' \vee \Gamma}, 
\end{equation}
for any $\Gamma, \Gamma' \in \mathcal{Y} ([2 d])$. 
Therefore, one may select a collection of independent generators as 
$\lbrace \wh{\hbar}_{\Gamma} \rbrace_{\Gamma \in \mathcal{Y}^{+} ([2 d])}$, where 
the set $\mathcal{Y}^{+} ([2 d])$ is recursively defined as follows: 
\begin{equation*} 
\begin{aligned}
\Gamma \in \mathcal{Y}^{+} ([2 d]), \quad &\text{if $|\Gamma| = 1$}, \\
\Gamma \in \mathcal{Y}^{+} ([2 d]), \quad &\text{if 
$L (\Gamma), R (\Gamma) \in \mathcal{Y}^{+} ([2 d]) \,\&\, L (\Gamma) \prec R (\Gamma)$}. 
\end{aligned}
\end{equation*}

\begin{definition}
The algebra 
$\mathcal{E}_{2 d}^{q} (\hbar) := K_{2 d}^{(q)} \langle \lbrace \xi_{\Gamma} \rbrace_{\Gamma \in \mathcal{Y}^{+} ([2 d])} 
\rangle/ \mathcal{I}_{2 d}^{(q)}$, where $\mathcal{I}_{2 d}^{(q)}$ is the ideal generated by the relations 
$[\xi_{\Gamma}, \xi_{\Gamma'}]_{q} = \xi_{\Gamma \vee \Gamma'}$, where 
$\Gamma$ and $\Gamma'$ vary over $\mathcal{Y}^{+} ([2 d])$, and $\Gamma \prec \Gamma'$, is termed 
\emph{the bracketing algebra}. 
The canonical images $\wh{\hbar}_{\Gamma} \in \mathcal{E}_{2 d}^{q} ([2 d])$ 
of the symbols $\xi_{\Gamma}$, $\Gamma \in \mathcal{Y}^{+} ([2 d])$, 
are termed \emph{the noncommutative Planck constants}. 
\end{definition}

\begin{remark}
Looking for a natural name for the algebra $\mathcal{E}_{2 d}^{q} ([2 d])$, 
we came up with the fact that \emph{bracketing} is also an important concept 
in the phenomenological philosophy of Edmund Husserl, which means something like 
``suspending of judgement'' that precedes a phenomenological analysis. 
It is also termed \emph{epoch\'e ($\varepsilon \pi \omega \chi \eta$)}, so perhaps 
another name for $\mathcal{E}_{2 d}^{q} ([2 d])$ could be the ``epoch\'e algebra''. 
\end{remark}

It is natural to keep the notational convention \eqref{eq:hbar_convention} for the noncommutative Planck constants 
$\wh{\hbar}_{\Gamma}$, $\Gamma \in \mathcal{Y}^{+} ([2 d])$, extending the notation $\wh{\hbar}_{\Gamma}$ 
to all leaf-labelled trees $\Gamma \in \mathcal{Y} ([2 d])$. 
The underlying vector space of the bracketing algebra is naturally graded as 
\begin{equation*} 
\mathrm{deg} \big( \wh{\hbar}_{\Gamma} \big) := |\Gamma| - 1, 
\end{equation*}  
where $|\Gamma|$ is the number of leaves in $\Gamma \in \mathcal{Y}^{+} ([2 d])$. 
Note, that $|\Gamma| - 1$ coincides with the number of internal vertices in $\Gamma$. 
The induced decreasing filtration respects the multiplication on $\mathcal{E}_{2 d}^{q} (\hbar)$, 
so one obtains an algebra filtration $\mathcal{F}^{\bullet} \mathcal{E}_{2 d}^{q} (\hbar)$, 
\begin{equation} 
\label{eq:epoche_filtration}
\mathcal{F}^{n} \mathcal{E}_{2 d}^{q} (\hbar) := \langle \lbrace \wh{\hbar}_{\Gamma} \,|\, 
\Gamma \in \mathcal{Y}^{+} ([2 d]), \mathrm{deg} (\wh{\hbar}_{\Gamma}) \geqslant n
\rbrace\rangle, 
\end{equation}
where $n = 0, 1, 2, \dots$. 

\begin{definition} 
The associated graded algebra $\mathcal{A}_{2 d}^{q} (\hbar) := \mathrm{gr}_{\mathcal{F}} \big( 
\mathcal{E}_{2 d}^{q} (\hbar) \big)$ with respect to filtration \eqref{eq:epoche_filtration} is termed 
\emph{the classical shadow} of the bracketing algebra. 
\end{definition}

We denote the canonical images of $\wh{\hbar}_{\Gamma}$ in the classical shadow as 
$\hbar_{\Gamma} \in \mathcal{A}_{2 d}^{q} (\hbar)$, $\Gamma \in \mathcal{Y} ([2 d])$ (i.e. without the hats). 
Observe, that $\mathcal{A}_{2 d}^{q} (\hbar)$ is nothing else but an infinite dimensional  $q$-affine space, 
\begin{equation*} 
\hbar_{\Gamma} \hbar_{\Gamma'} = q_{\Gamma', \Gamma} \hbar_{\Gamma'} \hbar_{\Gamma}, 
\end{equation*}
for $\Gamma, \Gamma' \in \mathcal{Y} ([2 d])$. 
A ``quantization'' should be perceived as a linear 
over $K_{2 d}^{(q)}$ map $Q : \mathcal{A}_{2 d}^{q} (\hbar) \to \mathcal{E}_{2 d}^{q} (\hbar)$. 
One obtains an example of such a map as follows. 
Denote 
$\mathbf{B} (\mathcal{A}_{2 d}^{q} (\hbar))$
the monomial basis in the classical shadow formed by the monomials of the shape 
$\hbar_{\Gamma_1} \hbar_{\Gamma_2} \dots \hbar_{\Gamma_n}$, where 
$n = \in \mathbb{Z}_{> 0}$, 
$\Gamma_1, \Gamma_2, \dots, \Gamma_{n} \in \mathcal{Y}([2 d])$, 
such that $\Gamma_1 \succcurlyeq \Gamma_2 \succcurlyeq \dots \succcurlyeq \Gamma_n$, 
the symbol $\succcurlyeq$ means ``$\succ$ or $=$''. 
For an element $w = \hbar_{\Gamma_1} \hbar_{\Gamma_2} \dots \hbar_{\Gamma_n} 
\in \mathbf{B} (\mathcal{A}_{2 d}^{q} (\hbar))$, set 
\begin{equation} 
\label{eq:normal_quantization}
Q : w \mapsto 
\wh{\hbar}_{\Gamma_1} \wh{\hbar}_{\Gamma_2} \dots \wh{\hbar}_{\Gamma_n}. 
\end{equation}
This extends by $K_{2 d}^{(q)}$-linearity to the whole $\mathcal{A}_{2 d}^{q} (\hbar)$, 
$Q: \mathcal{A}_{2 d}^{q} (\hbar) \to \mathcal{E}_{2 d}^{q} (\hbar)$. 
For a pair of basis elements $w, w' \in \mathbf{B} (\mathcal{A}_{2 d}^{q} (\hbar))$, 
define $q (w, w')$ by $w' w = q (w, w') w w'$. 
The commutation relations imply that 
$\Delta (w, w') := Q (w') Q (w) - q (w, w') 
Q (w) Q (w') \in \mathcal{F}^{|w| + |w'| + 1} \mathcal{E}_{2 d}^{q} (\hbar)$, where 
$|w|$ and $|w'|$ are the degrees of $w$ and $w'$, respectively. 
Furthermore, there exists a unique element $\langle w, w' \rangle \in \mathcal{A}_{2 d}^{q} (\hbar)$ of degree 
$|w| + |w'| + 1$, such that 
\begin{equation} 
\label{eq:normal_bracket}
\Delta (w, w') = 
Q (\langle w, w' \rangle) + Z, 
\end{equation}
where $Z \in \mathcal{F}^{|w| + |w'| + 2} \mathcal{E}_{2 d}^{q} (\hbar)$. 
Extending the notation $\langle -, - \rangle$ by $K_{2 d}^{(q)}$-bilinearity, one obtains a 
bracket $\langle -, - \rangle: \mathcal{A}_{2 d}^{q} (\hbar) \otimes \mathcal{A}_{2 d}^{q} (\hbar) \to 
\mathcal{A}_{2 d}^{q} (\hbar)$, where the tensor product is taken over $K_q$.

\begin{definition} 
The linear map $Q: \mathcal{A}_{2 d}^{q} (\hbar) \to \mathcal{E}_{2 d}^{q} (\hbar)$ defined
by \eqref{eq:normal_quantization} is termed 
\emph{the normal $q$-quantization} on the classical shadow. The bilinear map 
$\langle -, - \rangle: \mathcal{A}_{2 d}^{q} (\hbar) \otimes \mathcal{A}_{2 d}^{q} (\hbar) \to 
\mathcal{A}_{2 d}^{q} (\hbar)$
defined by \eqref{eq:normal_bracket} is termed \emph{the normal $q$-Poisson bracket} on the classical shadow. 
\end{definition}

\begin{remark}
The normal $q$-Poisson bracket is a graded map of degree $+1$. If we perceive the generators $\hbar_{\Gamma}$ 
corresponding to the trees with a single leaf, $|\Gamma| = 1$, as the analogues of the coordinates and momenta 
of a $d$-dimensional quantum mechanical system, then intuitively $\langle -, - \rangle$ 
corresponds to the Poisson bracket $\lbrace -, - \rbrace$ 
\emph{multiplied by the semiclassical parameter} $\hbar \to 0$. 
\end{remark}

It is worth to point out, that the normal $q$-quantization map \eqref{eq:normal_quantization} can be described using 
the calculus of functions of ordered operators  
\cite{Maslov_hyper,Maslov_char,Maslov_opmeth,Maslov_Karasev} 
(the $\mu$-structures). 
For a monomial $w = \hbar_{\Gamma_1} \hbar_{\Gamma_2} \dots \hbar_{\Gamma_n} 
\in \mathbf{B} (\mathcal{A}_{2 d}^{q} (\hbar))$, 
$\Gamma_1 \succcurlyeq \Gamma_2 \succcurlyeq \dots \succcurlyeq \Gamma_n$, 
one has: 
\begin{equation*} 
Q (w) = \Big\llbracket 
\overset{\nu_1}{\wh{\hbar}_{\Gamma_1}} \, 
\overset{\nu_2}{\wh{\hbar}_{\Gamma_2}} \, 
\dots 
\overset{\nu_n}{\wh{\hbar}_{\Gamma_n}}
\Big\rrbracket, 
\end{equation*}
where $\llbracket \dots \rrbracket$ denotes the \emph{autonomous bracket} \cite{Maslov_opmeth}, 
$\nu_i := \nu (\Gamma_i)$, $i = 1, 2, \dots, n$, and  
$\nu : \mathcal{Y} ([2 d]) \overset{\sim}{\to} \mathbb{Z}_{> 0}$ is the 
numbering described in the previous section. 
Basically, the indices $\nu_i$ indicate the ``order of action'' of the symbols $\wh{\hbar}_{\Gamma_i}$, , $i \in [n]$, 
and the symbol with the smallest number is put in the right most position in the product. 
For every $N = 1, 2, 3, \dots$, consider the space of smooth complex functions $S (\mathbb{R}^{N})$, which 
decay at infinity faster than any power of a polynomial (the Schwartz space). 
Suppose we have a noncommutative algebra $\mathcal{O}$ described in terms of generators 
$\xi_1, \xi_2, \dots, \xi_m$ and a finite number of noncommutative polynomial relations 
$R_1, R_2, \dots, R_k \in \mathbb{C} \langle \xi_1, \xi_2, \dots, \xi_m \rangle$, 
\begin{equation*} 
\mathcal{O} := \mathbb{C} \langle \xi_1, \xi_2, \dots, \xi_m \rangle/ (R_1, R_2, \dots, R_k).  
\end{equation*} 
Denote $A_i \in \mathcal{O}$ the canonical image of 
$\xi_i$, $i \in [m]$. 
If $f (x_1, x_2, \dots, x_N) \in S (\mathbb{R}^N)$ is a \emph{polynomial} function 
in $x_1, x_2, \dots, x_N$, then there is a 
well-defined notation 
\begin{equation*} 
\wh{f} = \big\llbracket 
f (\overset{\sigma (1)}{A_{i_1}}, \overset{\sigma (2)}{A_{i_2}}, \dots, 
\overset{\sigma (N)}{A_{i_N}}) \big\rrbracket \in \mathcal{O},  
\end{equation*}
for every permutation $\sigma \in S_N$, and every collection $i_1, i_2, \dots, i_N \in [m]$, 
where the indices atop correspond simply to the order of factors in the products. 
For example, if $N = 3$, and $f (x_1, x_2, x_3) = x_1 x_{2}^{5} x_{3}^{7}$, then 
$\wh{f}$ corresponding to a permutation 
$\sigma (1) = 2$, $\sigma (2) = 3$, $\sigma (3) = 1$ is going to be 
$A_{i_2}^{5} A_{i_1} A_{i_3}^{7}$, etc. 
A natural extension of this notation to the set $\mathcal{S} := \bigsqcup_{N = 1}^{\infty} S (\mathbb{R}^{N})$ is 
termed a \emph{$\mu$-structure} (for a complete list of axioms, see \cite{Maslov_opmeth,Maslov_Karasev}). 

In \cite{Maslov_Karasev} the authors consider a problem of quantization 
(the \emph{``asymptotic''} quantization) 
in a setting where 
one is given 
a family of defining relations 
$R_{1}^{(\varepsilon)}, R_{2}^{(\varepsilon)}, \dots, R_{k}^{(\varepsilon)} \in 
\mathbb{C} \langle \xi_1, \xi_2, \dots, \xi_m \rangle$ containing a ``small'' 
commutative parameter $\varepsilon \to 0$. 
This yields a family of noncommutative algebras $\mathcal{O}_{\varepsilon}$. 
It can happen, that $\mathcal{O}_{\varepsilon}$ admits a \emph{left regular representation}, i.e. 
for every $f (x_1, x_2, \dots, x_m) \in S (\mathbb{R}^{m})$ (recall, that $m$ is the number of 
generators $\xi_1, \xi_2, \dots, \xi_m$), 
and for every $j \in [m]$, there exists a unique $g_{\varepsilon} (x_1, x_2, \dots, x_m) \in S (\mathbb{R}^{m})$, such that
\begin{equation*}
\big\llbracket
\overset{j}{A_j^{(\varepsilon)}} f (\overset{1}{A_1^{(\varepsilon))}}, 
\overset{2}{A_2^{(\varepsilon)}}, \dots, 
\overset{m}{A_m^{(\varepsilon)}}) 
\big\rrbracket = 
\big\llbracket
g_{\varepsilon} (\overset{1}{A_1^{(\varepsilon)}}, 
\overset{2}{A_2^{(\varepsilon)}}, \dots, 
\overset{m}{A_m^{(\varepsilon)}}) 
\big\rrbracket, 
\end{equation*}
where $A_i^{(\varepsilon)}$ is the canonical image of $\xi_i$ in $\mathcal{O}_{\varepsilon}$, $i \in [m]$. 
If we assume that $\varepsilon$ is specialized to a particular value $\varepsilon \in [0, 1]$, then 
this defines the operators $\mathbb{L}_{j}^{(\varepsilon)}: S (\mathbb{R}^{m}) \to S (\mathbb{R}^{m})$, 
$f \mapsto g_{\varepsilon} = \mathbb{L}_j^{(\varepsilon)} (f)$, 
representing the generators $A_{j}^{(\varepsilon)} \in \mathcal{O}_{\varepsilon}$, $j \in [m]$. 
There is a star product $\star_{\varepsilon}$ on $S (\mathbb{R}^{m})$ defined by 
\begin{equation*} 
f \star_{\varepsilon} g := 
\big\llbracket 
f (\overset{1}{\mathbb{L}_{1}^{(\varepsilon)}}, 
\overset{2}{\mathbb{L}_{2}^{(\varepsilon)}}, \dots, 
\overset{m}{\mathbb{L}_{m}^{(\varepsilon)}})
\big\rrbracket g, 
\end{equation*}
for any $f, g \in S (\mathbb{R}^{m})$. 
A series of examples of such products is considered in \cite{Maslov_hyper}. 
According to the general philosophy advocated in \cite{Maslov_hyper,Maslov_Karasev}, 
one should define the \emph{generalized quantum Yang-Baxter equation} as a system of equations 
\begin{equation*} 
R_{j}^{(\varepsilon)} (\xi_i \to \mathbb{L}_{i}^{(\varepsilon)}, i \in [m]) = 0, \quad j \in [k], 
\end{equation*} 
where one replaces to symbols $\xi_i$ with the operators 
of the left regular representation $\mathbb{L}_{i}^{(\varepsilon)}$, $i \in [m]$, in 
the noncommutative 
polynomials $R_{j}^{(\varepsilon)} \in \mathbb{C} \langle \xi_1, \xi_2, \dots, \xi_m \rangle$, $j \in [k]$. 

Being applied to our case, we can connect the classical shadow $\mathcal{A}_{2 d}^{q} (\hbar)$ 
with the bracketing algebra $\mathcal{E}_{2 d}^{q} (\hbar)$ by a homotopy, modifying the relations 
for the noncommutative Planck constants as 
\begin{equation*} 
\wh{\hbar}_{\Gamma} \wh{\hbar}_{\Gamma'} - q_{\Gamma', \Gamma} \wh{\hbar}_{\Gamma'} \wh{\hbar}_{\Gamma} = 
\varepsilon \wh{\hbar}_{\Gamma \vee \Gamma'}, 
\end{equation*}
where $\varepsilon \in [0, 1]$, and $\Gamma, \Gamma' \in \mathcal{Y} ([2 d])$. 
If one reinterprets $\varepsilon$ and perceives it as a formal commutative parameter, i.e. as a central generator 
for a central extension of $\mathcal{E}_{2 d}^{q} (\hbar)$, then it becomes quite natural 
to consider the problem of quantization in terms of truncations of the Rees ring corresponding to the filtration 
$\mathcal{F}^{\bullet} \mathcal{E}_{2 d}^{q} (\hbar)$. 
In the next section we describe the corresponding star product $\star_{\varepsilon}$ explicitly, and this, as a side-effect, 
yields a left regular representation of the truncations $(\mathcal{E}_{2 d}^{q} (\hbar))_{\leqslant N}$, $N = 1, 2, \dots$, 
on the $q$-commutative polynomial rings.

\section{Normal noncommutative quantization}

In quantum mechanics of a $d$-dimensional system with 
canonical momenta $\wh{z}_1 = \wh{p}_1$, $\wh{z}_{2} = \wh{p}_2$, \dots $\wh{z}_{d} = \wh{p}_{d}$, and 
coordinates $\wh{z}_{d + 1} = \wh{x}_1$, $\wh{z}_{d + 2} = \wh{x}_{2}$, \dots $\wh{z}_{2 d} = \wh{x}_{d}$, 
if we have a pair of monomials 
\begin{equation*} 
\wh{f} = \Big\llbracket 
\overset{i_1}{\wh{z}_{i_1}} \overset{i_2}{\wh{z}_{i_2}} \dots \overset{i_m}{\wh{z}_{i_m}}
\Big\rrbracket, \quad 
\wh{g} = \Big\llbracket 
\overset{j_1}{\wh{z}_{j_1}} \overset{j_2}{\wh{z}_{j_2}} \dots \overset{j_n}{\wh{z}_{j_n}}
\Big\rrbracket, 
\end{equation*}
where $i_{\alpha}, j_{\beta} \in [2 d]$, $\alpha \in [m]$, $\beta \in [n]$, 
then their product $\wh{f} \wh{g}$ can again be expressed as 
\begin{equation*} 
\wh{f} \wh{g} = 
\Big\llbracket 
(f *_{\hbar} g) \Big( 
\overset{1}{\wh{z}_1}, 
\overset{2}{\wh{z}_2}, \dots 
\overset{2 d}{\wh{z}_{2 d}}  
\Big) 
\Big\rrbracket, 
\end{equation*}
where $(f *_{\hbar} g) (z_1, z_2, \dots, z_{2 d}) \in \mathbb{C} [z_1, z_2, \dots, z_{2 d}] [\hbar]$, 
and $\hbar$ is the Planck constant, 
\begin{equation*} 
[\wh{z}_{i}, \wh{z}_{j}] = - \ii \hbar (\delta_{i, j - d} - \delta_{i - d, j}), 
\end{equation*} 
where $i, j \in [2 d]$, $\hbar$ is central, and $\delta$ is the Kronecker symbol. 
Extending by $\mathbb{C} [\hbar]$-linearity, one obtains a noncommutative product on 
$\mathbb{C} [z_1, z_2, \dots, z_{2 d}] [\hbar]$, 
which is most easily described in terms of the \emph{Wick contractions}, 
\begin{equation} 
\label{eq:Wick_contraction}
\acontraction{}{z}{_i}{z}
z_{i} z_{j} := - \ii \hbar \delta_{i, j - d}, 
\end{equation}
for $i, j \in [2 d]$. 
For the monomials $f 
= z_{i_1} z_{i_2} \dots, z_{i_m}$ and 
$g 
= z_{j_1} z_{j_2} \dots z_{j_n}$
as above, we have: 
\begin{equation*} 
f *_{\hbar} g = \sum_{\alpha = 1}^{m} \sum_{\beta = 1}^{n} 
\acontraction{}{z}{_{i_{\alpha}} }{z}
z_{i_{\alpha}} z_{j_{\beta}} \, 
(z_{i_1} \dots \check z_{i_{\alpha}} \dots z_{i_m}) \, 
(z_{j_1} \dots \check z_{j_{\beta}} \dots z_{j_n}), 
\end{equation*}
where the check mark atop denotes that the corresponding symbol in the product is omitted. 

In what follows, it is natural to reinterpret the sum over Wick contractions as 
a sum over leaf-labelled trees of degree two (i.e. the symbol of the contraction 
can be perceived as an unlabelled planar binary tree with two leaves). 
Consider the epoch\'e algebra $\mathcal{E}_{2 d}^{q} (\hbar)$ around the semiclassical parameter $\hbar$. 
Recall, that we have defined a numbering $\nu: \mathcal{Y} ([2 d]) \overset{\sim}{\to} \mathbb{Z}_{> 0}$ 
on the collection of \emph{all} leaf-labelled trees, but we can still use it to define the ordering of factors 
in the products of the generators $\wh{\hbar}_{\Gamma}$, $\Gamma \in \mathcal{Y}^{+} ([2 d])$. 
Denote $\mathbf{B} (\mathcal{E}_{2 d}^{q} (\hbar))$ a set of all monomials of the shape 
\begin{equation*} 
\wh{w}_{\Gamma_1, \Gamma_2, \dots, \Gamma_m} := 
\wh{\hbar}_{\Gamma_1}
\wh{\hbar}_{\Gamma_2}
\dots 
\wh{\hbar}_{\Gamma_m}, 
\end{equation*}
where $\Gamma_1 \succcurlyeq \Gamma_2 \succcurlyeq \dots \succcurlyeq \Gamma_m$, 
$\Gamma_i \in \mathcal{Y}^{+} ([2 d])$, $i \in [m]$, $m \in \mathbb{Z}_{> 0}$. 

\begin{proposition}
The set $\mathbf{B} (\mathcal{E}_{2 d}^{q} (\hbar))$ is a basis of the underlying vector space of the 
epoch\'e algebra $\mathcal{E}_{2 d}^{q} (\hbar)$. 
\end{proposition}

\begin{proof}
The fact claimed is a straightforward consequence of the ``diamond lemma'' \cite{Newmann}, based 
on the observation, that $\nu ([\wh{\hbar}_{\Gamma}, \wh{\hbar}_{\Gamma'}]_{q})$ is always greater that 
$\nu (\wh{\hbar}_{\Gamma})$ and $\nu (\wh{\hbar}_{\Gamma'})$, where $\Gamma, \Gamma' \in \mathcal{Y}^{+} ([2 d])$, 
and the observation, that the length of a monomial $\wh{w} [\wh{\hbar}_{\Gamma}, \wh{\hbar}_{\Gamma'}]_{q} \wh{w'}$, 
where $\wh{w}, \wh{w'} \in \mathbf{B} (\mathcal{E}_{2 d}^{q} (\hbar))$, is smaller that the length of 
$\wh{w} \wh{\hbar}_{\Gamma} \wh{\hbar}_{\Gamma'} \wh{w'}$. 
\qed
\end{proof}

Recall now that we have a decreasing filtration $\mathcal{F}^{\bullet} \mathcal{E}_{2 d}^{q} (\hbar)$ induced by the 
grading $\mathrm{deg} (\wh{\hbar}_{\Gamma}) = |\Gamma| - 1$ 
of the vector space spanned over all noncommutative Planck constants $\wh{\hbar}_{\Gamma}$, $\Gamma \in \mathcal{Y} ([2 d])$.  
The associated graded $\mathcal{A}_{2 d}^{q} (\hbar)$ (i.e. the classical shadow) has a basis $\mathbf{B} (\mathcal{A}_{2 d}^{q} (\hbar))$ 
formed by the monomials 
\begin{equation} 
\label{eq:basis_classical_shadow}
w_{\Gamma_1, \Gamma_2, \dots, \Gamma_m} := 
\hbar_{\Gamma_1}
\hbar_{\Gamma_2}
\dots 
\hbar_{\Gamma_m}, 
\end{equation}
where $\Gamma_1 \succcurlyeq \Gamma_2 \succcurlyeq \dots \succcurlyeq \Gamma_m$, 
$\Gamma_i \in \mathcal{Y}^{+} ([2 d])$, $i \in [m]$, $m \in \mathbb{Z}_{> 0}$, 
and $\hbar_{\Gamma}$ is the canonical image of $\wh{\hbar}_{\Gamma}$, $\Gamma \in \mathcal{Y}^{+} ([2 d])$. 
This yields a \emph{vector space} isomorphism 
\begin{equation} 
\label{eq:varphi_iso}
\varphi_{q}: w_{\Gamma_1, \Gamma_2, \dots, \Gamma_m} \mapsto 
\wh{w}_{\Gamma_1, \Gamma_2, \dots, \Gamma_m} 
\end{equation}
between the underlying vector spaces of the 
classical shadow $\mathcal{A}_{2 d}^{q} (\hbar)$ and 
the epoch\'e algebra $\mathcal{E}_{2 d}^{q} (\hbar)$.  
If we take $\wh{u}, \wh{v} \in \mathbf{B} (\mathcal{E}_{2 d}^{q} (\hbar))$, then the basis property implies, that 
there exists a \emph{unique} function $c_{u, v}: \mathbf{B} (\mathcal{E}_{2 d}^{q} (\hbar)) \to \mathbb{C}$ with a 
\emph{finite} support $\Omega$, such that $\wh{u} \wh{v} = \sum_{\wh{w} \in \Omega} c_{u, v} (\wh{w}) \wh{w}$. 
Therefore, we can induce a noncommutative product $\star$ on the classical shadow $\mathcal{A}_{2 d}^{q} (\hbar)$ as follows: 
\begin{equation*} 
\varphi_{q}^{-1} (\wh{u}) \star \varphi_{q}^{-1} (\wh{v}) = 
\sum_{\wh{w} \in \Omega} c_{u, v} (\wh{w}) \varphi_{q}^{- 1} (\wh{w}),  
\end{equation*}
imposing bilinearity of $\star$ over $K_{2 d}^{(q)}$. 
This star product should be perceived as a natural analogue of the star product $*_{\hbar}$ in quantum mechanics. 
Of course, if we had considered a central extension of $\mathcal{E}_{2 d}^{q} (\hbar)$ by a formal central generator $\varepsilon$, 
$[\wh{\hbar}_{\Gamma}, \wh{\hbar}_{\Gamma'}]_{q} = \varepsilon \wh{\hbar}_{\Gamma \vee \Gamma'}$, we would had obtained 
a product $\star_{\varepsilon}$ on the extension $\mathcal{A}_{2 d}^{q} (\hbar) \otimes \mathbb{C} [\varepsilon]$. 
In this sense, $\varepsilon$ corresponds to a semiclassical parameter $\hbar \to 0$ in quantum mechanics. 
 
To describe the star product $\star : \mathcal{A}_{2 d}^{q} (\hbar) \otimes \mathcal{A}_{2 d}^{q} (\hbar) \to 
\mathcal{A}_{2 d}^{q} (\hbar)$ explicitly, we need some notation. 
In place of the Wick contractions \eqref{eq:Wick_contraction}, define 
\begin{equation} 
\label{eq:hbar_plus}
\hbar_{\Gamma}^{+} := 
\begin{cases}
\hbar_{\Gamma}, &\text{if $\Gamma \in \mathcal{Y}^{+} ([2 d])$},\\ 
0, &\text{otherwise}, 
\end{cases}
\end{equation}
for any $\Gamma \in \mathcal{Y} ([2 d])$. 
Recall, that if we have an unlabelled planar binary tree $T \in \mathcal{Y}$ with $n$ leaves, $|T| = n$, 
then we have 
a notation $T (\Gamma_1, \Gamma_2, \dots, \Gamma_n) \in \mathcal{Y} ([2 d])$, 
for any $\Gamma_{i} \in \mathcal{Y} ([2 d])$, $i \in [n]$. 
The number of leaves in $\Gamma = T (\Gamma_1, \Gamma_2, \dots, \Gamma_n)$ equals $|\Gamma| = |\Gamma_1| + |\Gamma_2| + \dots + |\Gamma_{n}|$, 
and the labelling is inherited from the labellings of the arguments. 
For a finite sequence $\wt{\Gamma} = (\Gamma_1, \Gamma_2, \dots, \Gamma_{n})$ of trees $\Gamma_{i} \in \mathcal{Y} ([2 d])$, $i \in [n]$, 
and a permutation $\sigma \in S_{n}$, set 
\begin{equation*}
\wt{\Gamma}_{\sigma} := (\Gamma_{\sigma (1)}, \Gamma_{\sigma (2)}, \dots, \Gamma_{\sigma (n)}). 
\end{equation*}
Define a coefficient $q (\wt{\Gamma}, \sigma)$ by 
\begin{equation} 
\label{eq:q_Gamma_tilde_sigma}
\hbar_{\Gamma_{\sigma (1)}} \hbar_{\Gamma_{\sigma (2)}} \dots \hbar_{\Gamma_{\sigma (n)}} = 
q (\wt{\Gamma}, \sigma) 
\hbar_{\Gamma_1} \hbar_{\Gamma_2} \dots \hbar_{\Gamma_n}. 
\end{equation}
In particular, $q (\wt{\Gamma}, \mathit{id}) = 1$. 
Note that $q (\wt{\Gamma}, \sigma)$ is just some product of the entries of the matrix $q = \| q_{i, j} \|_{i, j \in [2 d]}$, and 
observe that 
\begin{equation} 
\label{eq:q_Gamma_tilde_sigma_tau}
q (\wt{\Gamma}, \sigma \circ \tau) = q (\wt{\Gamma}, \sigma) q (\wt{\Gamma}_{\sigma}, \tau), 
\end{equation}
for any $\sigma, \tau \in S_n$. Therefore, $q (\wt{\Gamma}, \sigma)^{-1} = q (\wt{\Gamma}_{\sigma}, \sigma^{-1})$. 
For a vector of positive integers $(m_1, m_2, \dots, m_p)$, such that $m_1 + m_2 + \dots + m_p = n$, denote 
\begin{equation*} 
S_{n}^{(m_1, m_2, \dots, m_p)} := \lbrace \sigma \in S_{n} \,|\, 
\sigma (l_j + 1) < \dots < \sigma (l_j + m_j), j \in [p]
\rbrace, 
\end{equation*}
where $l_j := \sum_{i = 1}^{j - 1} m_i$, $j \in [p]$. 
We need also a notation: 
\begin{equation} 
\label{eq:theta_Gamma_tilde}
\theta_{\wt{\Gamma}} \equiv 
\theta_{\Gamma_1, \Gamma_2, \dots, \Gamma_m}
:= 
\begin{cases} 
1, \quad &\text{if $\Gamma_1 \succcurlyeq \Gamma_2 \succcurlyeq \dots \succcurlyeq \Gamma_m$},\\ 
0, \quad &\text{otherwise}, 
\end{cases}
\end{equation}
where $\wt{\Gamma} = (\Gamma_1, \Gamma_2, \dots, \Gamma_m)$, 
$\Gamma_i \in \mathcal{Y} ([2 d])$, $i \in [m]$.

\begin{theorem}
The normal star product $\star$ of a pair of elements of 
$\mathbf{B} (\mathcal{A}_{2 d}^{q} (\hbar))$ in the notation 
\eqref{eq:basis_classical_shadow}, \eqref{eq:hbar_plus}, \eqref{eq:theta_Gamma_tilde} 
is defined by the formula  
\begin{multline}
\label{eq:normal_star_product} 
w_{\Gamma_1, \Gamma_2, \dots, \Gamma_m} \star w_{\Gamma_{m + 1}, \Gamma_{m + 2}, \dots, \Gamma_{N}} = 
\sum_{p = 1}^{N} \sum_{\substack{T_1, T_2, \dots, T_p \in \mathcal{Y},\\|T_1| + |T_2| + \dots + |T_p| = N}}
\times \\ \times
\sum_{\sigma \in S_{N}^{(|T_1|, |T_2|, \dots, |T_p|)}} 
(q (\wt{\Gamma}, \sigma))^{-1} \, 
\theta_{\wt{G} (\wt{T}, \wt{\Gamma}, \sigma)} \, 
\hbar_{G_1 (\wt{T}, \wt{\Gamma}, \sigma)}^{+}
\hbar_{G_2 (\wt{T}, \wt{\Gamma}, \sigma)}^{+} 
\dots 
\hbar_{G_p (\wt{T}, \wt{\Gamma}, \sigma)}^{+}, 
\end{multline}
where $\wt{\Gamma} = (\Gamma_1, \Gamma_2, \dots, \Gamma_N) \in (\mathcal{Y}^{+} ([2 d]))^{N}$, 
$\wt{T} = (T_1, T_2, \dots, T_p) \in \mathcal{Y}^{N}$, 
and 
$\wt{G} (\wt{\Gamma}, \wt{T}, \sigma) = 
(G_1 (\wt{\Gamma}, \wt{T}, \sigma), 
G_2 (\wt{\Gamma}, \wt{T}, \sigma), 
\dots, 
G_n (\wt{\Gamma}, \wt{T}, \sigma))$,  
\begin{equation*} 
G_j (\wt{\Gamma}, \wt{T}, \sigma) := 
T_j (\Gamma_{\sigma (l_j (\wt{T}) + 1)}, 
\Gamma_{\sigma (l_j (\wt{T}) + 2)}, \dots 
\Gamma_{\sigma (l_j (\wt{T}) + |T_j|)} ), 
\end{equation*}
where $l_j (\wt{T}) = \sum_{i = 1}^{j - 1} |T_i|$, for $j \in [p]$. 
\end{theorem}

\begin{proof}
The formula \eqref{eq:normal_star_product} can be proven by induction. 
Basically it says, that we should consider a collection $\wt{\Gamma} = (\Gamma_{1}, \Gamma_{2}, \dots, \Gamma_{N})$ 
of the leaf labelled trees $\Gamma_i$, $i \in [N]$, (where $\Gamma_{i} \succcurlyeq \Gamma_{i + 1}$, unless $i = m$), 
split it in all possible ways into $p$ ordered subsets of the sizes $m_1, m_2, \dots, m_p > 0$, $m_1 + m_2 + \dots + m_p = N$, 
where $p$ varies over $[N]$, and then span all possible trees $T_1, T_2, \dots, T_p$ over these groups of arguments, 
$|T_i| = m_i$, $i \in [p]$. The notation of the shape $\theta_{\wt{\Gamma}}$ 
defined in \eqref{eq:theta_Gamma_tilde} and the notation 
$\hbar_{\Gamma}^{+}$ from \eqref{eq:hbar_plus} select the terms which appear after the process of reordering of factors 
using the commutation relations $\hbar_{\Gamma} \star \hbar_{\Gamma'} - q_{\Gamma', \Gamma} \hbar_{\Gamma'} \star \hbar_{\Gamma} = 
\hbar_{\Gamma \vee \Gamma'}$, and $q (\wt{\Gamma}, \sigma)^{-1}$ is the corresponding braiding coefficient. 
\qed
\end{proof}

\begin{proposition} 
The star product of a pair of elements of the basis $\mathbf{B} (\mathcal{A}_{2 d}^{q} (\hbar))$ 
in the notation \eqref{eq:basis_classical_shadow} can be expressed as follows: 
\begin{equation} 
\label{eq:w_star_hbar_star}
w_{\Gamma_1, \Gamma_2, \dots, \Gamma_m} \star w_{\Gamma_{m + 1}, \Gamma_{m + 2}, \dots, \Gamma_{N}} = 
\hbar_{\Gamma_1} \star \hbar_{\Gamma_2} \star \dots \star \hbar_{\Gamma_N}.  
\end{equation}
\end{proposition}

\begin{proof}
One needs to apply the vector space map 
$\varphi_{q}: \mathcal{A}_{2 d}^{q} (\hbar) \overset{\sim}{\to} \mathcal{E}_{2 d}^{q} (\hbar)$ 
defined by \eqref{eq:varphi_iso} 
to the left and the right-hand sides of \eqref{eq:w_star_hbar_star}, 
to expand the definitions \eqref{eq:basis_classical_shadow} of the basis elements on the left hand side, 
and to use the fact, that 
$\varphi_{q}$ is an \emph{algebra} isomorphism 
$(\mathcal{A}_{2 d}^{q} (\hbar), \star) \overset{\sim}{\to} \mathcal{E}_{2 d}^{q} (\hbar)$, 
by construction. 
\qed
\end{proof}

The normal star product $\star$ respects the decreasing filtration $\mathcal{F}^{\bullet} (\mathcal{A}_{2 d}^{q} (\hbar), \star)$ 
induced by the grading $\mathrm{deg} (\hbar_{\Gamma}) = |\Gamma| - 1$, $|\Gamma|$ is the number of leaves in $\Gamma \in \mathcal{Y}^{+} ([2 d])$. 
The underlying vector space of $\mathcal{F}^{n} (\mathcal{A}_{2 d}^{q} (\hbar), \star)$, $n = 0, 1, 2 \dots$, is 
spanned by the products $\hbar_{\Gamma_1} \star \hbar_{\Gamma_2} \star \dots \star \hbar_{\Gamma_p}$, 
where $p = 1, 2, \dots$, and $|\Gamma_1| + |\Gamma_2| + \dots + |\Gamma_p| - p = n$, $\Gamma_i \in \mathcal{Y}^{+} ([2 d])$, $i \in [p]$. 
The algebras $(\mathcal{A}_{2 d}^{q} (\hbar), \star)$ and $\mathcal{E}_{2 d}^{q} (\hbar)$ are canonically isomorphic as \emph{filtered} algebras 
(by construction), and the classical shadow $\mathcal{A}_{2 d}^{q} (\hbar)$ can therefore be identified with the 
associated graded of $\mathcal{F}^{n} (\mathcal{A}_{2 d}^{q} (\hbar), \star)$. 
Since taking a $q$-commutator increases the filtration degree by one, if we look at it in the classical shadow, 
this yields a graded bilinear map $\langle -, - \rangle$ (a bracket on $\mathcal{A}_{2 d}^{q} (\hbar)$) of degree $+1$. 

\begin{proposition} 
The $q$-commutator in $\mathcal{F}^{\bullet} (\mathcal{A}_{2 d}^{q} (\hbar), \star)$ induces a graded bilinear map of degree $+1$ on 
the classical shadow $\mathcal{A}_{2 d}^{q} (\hbar)$, 
\begin{equation*} 
\langle -, - \rangle: \mathcal{A}_{2 d}^{q} (\hbar) \otimes \mathcal{A}_{2 d}^{q} (\hbar) \to \mathcal{A}_{2 d}^{q} (\hbar), 
\end{equation*} 
which satisfies the $q$-Poisson bracket axioms. 
\end{proposition}

\begin{proof}
Take a pair of basis elements $w = w_{\Gamma_1, \Gamma_2, \dots, \Gamma_m} \in \mathbf{B} (\mathcal{A}_{2 d}^{q} (\hbar))$ and 
$w' = w_{\Gamma_1', \Gamma_2', \dots, \Gamma_n'} \in \mathbf{B} (\mathcal{A}_{2 d}^{q} (\hbar))$. 
Their degrees are as follows: 
$|w| \equiv \mathrm{deg} (w) = \sum_{i = 1}^{m} |\Gamma_i| - m$, and 
$|w'| \equiv \mathrm{deg} (w') = \sum_{j = 1}^{n} |\Gamma_j'| - n$. 
Recall, that this notation implies 
$\Gamma_1 \succcurlyeq \Gamma_2 \succcurlyeq \dots \succcurlyeq \Gamma_m$, and 
$\Gamma_1' \succcurlyeq \Gamma_2' \succcurlyeq \dots \succcurlyeq \Gamma_n'$. 
From the explicit formula \eqref{eq:normal_star_product} for the star product, we have: 
\begin{multline*} 
w \star w' = 
w w' + 
\sum_{i = 1}^{m} \sum_{j = 1}^{n} 
\big( q \big( \wt{\Gamma} \vee \wt{\Gamma}', \sigma_{m, n}^{(i, j)} \big) \big)^{-1} 
\hbar_{\Gamma_i \vee \Gamma_j'}^{+} 
(\hbar_{\Gamma_m} \dots \hbar_{\Gamma_{i + 1}} 
\times \\ \times 
\hbar_{\Gamma_{i - 1}} \dots \hbar_{\Gamma_1}) \,
(\hbar_{\Gamma_n'} \dots \hbar_{\Gamma_{j + 1}'} \hbar_{\Gamma_{j - 1}'} \dots \hbar_{\Gamma_1'}) + 
O_{|w| + |w'| + 2}, 
\end{multline*} 
where $O_{|w| + |w'| + 2}$ stands for an element of filtration degree $|w| + |w'| + 2$, 
$\wt{\Gamma} \vee \wt{\Gamma'}$ denotes the concatenation of the lists 
$\wt{\Gamma} = (\Gamma_m, \dots, \Gamma_2, \Gamma_1)$ and 
$\wt{\Gamma}' = (\Gamma_n', \dots, \Gamma_2', \Gamma_1')$, 
and $\sigma_{m, n}^{(i, j)} \in S_{m + n}$ denotes a permutation 
\begin{multline*} 
(\sigma_{m, n}^{(i, j)} (1), \sigma_{m, n}^{(i, j)} (2), \dots, \sigma_{m, n}^{(i, j)} (m + n)) := 
(i, j, \, 1, \dots, i - 1, i + 1, \dots, m, \\ 
m + 1, \dots, m + j - 1, m + j + 1, \dots, m + n). 
\end{multline*}
The opposite star product $w' \star w$ looks totally similar. 
Taking the $q$-commutator $[w, w']_q^{\star} := w \star w' - q (w', w) w' \star w$, where 
$q (w', w)$ is determined from $w w' = q (w', w) w' w$, 
one extracts $\langle w, w' \rangle$ as the element of degree $|w| + |w'| + 1$ satisfying 
$[w, w']_{q}^{\star} = \langle w, w' \rangle + O_{|w| + |w'| + 2}$. 
It is straightforward to check, that 
\begin{equation*} 
\langle w, w' \rangle = - q (w', w) \langle w', w \rangle, 
\end{equation*}
i.e. the first axiom of a $q$-Poisson bracket (the $q$-antisymmetry) is satisfied. 
If we take another arbitrary element $w'' = w_{\Gamma_1'', \Gamma_2'', \dots, \Gamma_k''} 
\in \mathbf{B} (\mathcal{A}_{2 d}^{q} (\hbar))$, then one can establish 
the other two axioms, i.e. the $q$-Leibniz rule 
\begin{equation*}
\langle w, w' w'' \rangle = 
\langle w, w' \rangle w'' + 
q (w, w') w' \langle w, w'' \rangle,  
\end{equation*}
and the $q$-Jacobi identity 
\begin{equation*} 
\langle w, \langle w', w'' \rangle \rangle = 
\langle \langle w, w' \rangle, w'' \rangle + 
q (w, w')
\langle w', \langle w, w'' \rangle \rangle, 
\end{equation*} 
by a straightforward computation. 
Observe, that with these properties, the bracket $\langle -, - \rangle$ 
is determined by its values on the generators $\hbar_{\Gamma}$, which are just 
$\langle \hbar_{\Gamma}, \hbar_{\Gamma'} \rangle = \hbar_{\Gamma \vee \Gamma'}$, 
for $\Gamma, \Gamma' \in \mathcal{Y}^{+} ([2 d])$. 
\qed 
\end{proof}

\section{Distortion of the Weyl quantization}

We are now interested in a $q$-analogue of the Weyl quantization. 
In quantum mechanics of a $d$-dimensional system, this is basically a symmetrization map 
$W : \mathcal{A} \to \mathcal{B}$, where $\mathcal{A} = \mathbb{C} [z_1, z_2, \dots, z_{2d}] \otimes \mathbb{C} [h]$, 
and $\mathcal{B} = \mathbb{C} \langle \xi_1, \xi_2, \dots, \xi_{2 d} \rangle [\eta] / \mathcal{I}$, 
where $\mathcal{I}$ is the ideal generated by the canonical commutation relations 
$[\xi_i, \xi_{d + j}]= - \ii \eta$, $[\xi_i, \xi_j] = 0$, and 
$[\xi_{d + i}, \xi_{d + j}] = 0$, where $[-, -]$ denotes a commutator, and $i, j \in [d]$. 
The canonical images of $\xi_1, \xi_2, \dots, \xi_d$ in $\mathcal{B}$ are denoted as 
$\wh{z}_1 = \wh{p}_1, \wh{z}_2 = \wh{p}_2, \dots, \wh{z}_d = \wh{p}_{d}$ 
(the canonical momenta), the canonical images of $\xi_{d + 1}, \xi_{d + 2}, \dots, \xi_{2 d}$ are denoted as 
$\wh{z}_{d + 1} = \wh{x}_{1}, 
\wh{z}_{d + 2} = \wh{x}_{2}, \dots, 
\wh{z}_{2 d} = \wh{x}_{d}$ (the coordinates), and the canonical image of $\eta$ is denoted $\hbar$ 
(the Planck constant). 
The \emph{Weyl quantization map} is a linear map $W: \mathcal{A} \to \mathcal{B}$ defined 
on the monomials in generators as follows: 
\begin{equation*} 
W: 
h^{m} z_{i_1} z_{i_2} \dots z_{i_n} \mapsto \frac{\hbar^{m}}{n!} 
\sum_{\sigma \in S_n} \wh{z}_{i_{\sigma (1)}} \wh{z}_{i_{\sigma (2)}} \dots \wh{z}_{i_{\sigma (n)}}, 
\end{equation*}
where $m \in \mathbb{Z}_{\geqslant 0}$, $i_{\alpha} \in [2 d]$, $\alpha \in [n]$, $n \in \mathbb{Z}_{> 0}$. 
A characteristic property of the Weyl quantization which selects $W$ from the other possible quantizations 
(like, for example, the normal quantization) is the \emph{affine equivariance}. 
If we take an arbitrary $2 d \times 2 d$ matrix $A$ over $\mathbb{C}$, then we can act with it on the 
column of the generators $(z_1, z_2, \dots, z_{2 d})^{T}$, or on the column 
of their quantized analogues $(\wh{z}_1, \wh{z}_2, \dots, \wh{z}_{2 d})^{T}$, where $(-)^{T}$ denotes transposition. 
One can first act, and then quantize, or first quantize, and then act. 
The affine equivariance is the property that the two results coincide for any $A$: 
\begin{equation} 
\label{eq:affine_Weyl}
W ( (A z)_{i_1} (A z)_{i_2} \dots (A z)_{i_n} ) = 
\frac{1}{n!} \sum_{\sigma \in S_{n}} 
(A \wh{z})_{i_{\sigma (1)}} (A \wh{z})_{i_{\sigma (2)}} \dots (A \wh{z})_{i_{\sigma (n)}},  
\end{equation}  
where $(A z)_{i} := \sum_{j = 1}^{2 d} A_{i, j} z_j$, and $(A \wh{z})_{i} := \sum_{j = 1}^{2 d} A_{i, j} \wh{z}_{j}$, $i \in [2 d]$. 

We would like to define a $q$-Weyl quantization map as a linear map $W^{(q)}: \mathcal{A}_{2 d}^{q} (\hbar) \to \mathcal{E}_{2 d}^{q} (\hbar)$, 
where $\mathcal{E}_{2 d}^{q} (\hbar)$ is the epoch\'e algebra discussed in the previous sections, and 
$\mathcal{A}_{2 d}^{q} (\hbar)$ is its classical shadow. Suppose we wish to define it by the formula 
\begin{equation*} 
W^{(q)}: \hbar_{\Gamma_1} \hbar_{\Gamma_2} \dots \hbar_{\Gamma_{n}} \mapsto 
\sum_{\sigma \in S_n} C_{\Gamma_1, \Gamma_2, \dots, \Gamma_n}^{(q)} (\sigma)
\wh{\hbar}_{\Gamma_{\sigma (1)}} \wh{\hbar}_{\Gamma_{\sigma (2)}} \dots \wh{\hbar}_{\Gamma_{\sigma (n)}}, 
\end{equation*}
where $C_{\Gamma_1, \Gamma_2, \dots, \Gamma_n}^{(q)} (\sigma)$ are some coefficients, and 
$\Gamma_i \in \mathcal{Y}^{+} ([2 d])$, $i \in [n]$, $n \in \mathbb{Z}_{> 0}$, 
$\Gamma_1  \succcurlyeq \Gamma_2 \succcurlyeq \dots \succcurlyeq \Gamma_n$. 
Then a natural condition that 
\begin{equation*} 
W^{(q)} (\hbar_{\Gamma_1} \hbar_{\Gamma_2} \dots \hbar_{\Gamma_{n}}) - 
\wh{\hbar}_{\Gamma_1} \wh{\hbar}_{\Gamma_2} \dots \wh{\hbar}_{\Gamma_{n}} \in \mathcal{F}^{p + 1} \mathcal{E}_{2 d}^{q} (\hbar), 
\end{equation*}
where $p = |\Gamma_1| + |\Gamma_2| + \dots + |\Gamma_{n}| - n$, and 
$\mathcal{F}^{\bullet} \mathcal{E}_{2 d}^{q} (\hbar)$ is the filtration on $\mathcal{E}_{2 d}^{q} (\hbar)$ described 
in the previous sections, yields a condition on the coefficients 
\begin{equation*} 
\sum_{\sigma \in S_n} C_{\Gamma_1, \Gamma_2, \dots, \Gamma_n}^{(q)} (\sigma) q (\wt{\Gamma}, \sigma) = 1, 
\end{equation*}
where $\wt{\Gamma} = (\Gamma_1, \Gamma_2, \dots, \Gamma_n)$, and one uses 
the notation \eqref{eq:q_Gamma_tilde_sigma}. 
The problem remains: what is the correct way to define the 
coefficients $C_{\Gamma_1, \Gamma_2, \dots, \Gamma_n}^{(q)}$ satisfying this condition? 

It is natural to generalize the affine equivariance in the quantum 
mechanical case \eqref{eq:affine_Weyl} 
into the \emph{left and right affine coequivariance}. 
It is convenient to consider the truncations $(\mathcal{E}_{2 d}^{q} (\hbar))_{\leqslant N}$ and 
$(\mathcal{A}_{2 d}^{q} (\hbar))_{\leqslant N}$, for any $N \in \mathbb{Z}_{> 0}$, since the corresponding sums 
are going to be in this case finite. 
Formally, these truncations can be perceived as setting all the generators 
$\hbar_{\Gamma}$ and $\wh{\hbar}_{\Gamma}$ to zero, if $|\Gamma| > N$, $\Gamma \in \mathcal{Y}^{+} ([2 d])$. 
Let 
\begin{equation*}
\mathcal{Y}_{\leqslant N}^{+} ([2 d]) := \lbrace \Gamma \in \mathcal{Y}^{+} ([2 d]) \,|\, |\Gamma| \leqslant N \rbrace, 
\end{equation*} 
and denote $\hbar_{\Gamma}^{(N)}$ the canonical image of $\hbar_{\Gamma}$, $|\Gamma| \leqslant N$, in the truncation 
$(\mathcal{A}_{2 d}^{q} (\hbar))_{\leqslant N}$ of the classical shadow, $\Gamma \in \mathcal{Y}_{\leqslant N}^{+} ([2 d])$, 
and denote $\wh{\hbar}_{\Gamma}^{(N)}$ the canonical image of $\wh{\hbar}_{\Gamma}$, $|\Gamma| \leqslant N$, 
in the truncation $(\mathcal{E}_{2 d}^{q} (\hbar))_{\leqslant N}$. 

Let $\mathcal{M}_{2 d}^{q}$ be an algebra of the shape 
\begin{equation}
\label{eq:algebra_M}
\mathcal{M}_{2 d}^{q} = \mathbb{C} \big\langle \lbrace \Lambda_{\Gamma, \Gamma'}
\rbrace_{\Gamma, \Gamma' \in \mathcal{Y}^{+} ([2 d])} 
\big\rangle/ \mathcal{I}, 
\end{equation}
where $\Lambda_{\Gamma, \Gamma'}$ denote the generators written in a matrix form, and 
$\mathcal{I}$ is an ideal generated by a countable collection of noncommutative polynomials. 
Denote 
$A_{\Gamma, \Gamma'}$
the canonical image of $\Lambda_{\Gamma, \Gamma'}$ in $\mathcal{M}_{2 d}^{q}$. 
The \emph{left affine coequivariance condition} is the following equality in 
$\mathcal{M}_{2 d}^{q} \otimes (\mathcal{E}_{2 d}^{q} (\hbar))_{\leqslant N}$: 
\begin{multline}
\label{eq:left_affine_coequivariance} 
\sum_{\substack{\Gamma_1', 
\dots, \Gamma_n' \in \mathcal{Y}_{\leqslant N}^{+} ([2 d]),\\ \sigma \in S_{n}}}
A_{\Gamma_1, \Gamma_1'} 
\dots 
A_{\Gamma_n, \Gamma_n'} \otimes 
C_{\Gamma_1', 
\dots, \Gamma_n'}^{(q, N)} (\sigma) 
\wh{\hbar}_{\Gamma_{\sigma (1)}'}^{(N)}
\dots 
\wh{\hbar}_{\Gamma_{\sigma (n)}'}^{(N)}
= \\ = 
\sum_{\substack{\sigma \in S_{n}, \\ 
\Gamma_1', \dots, \Gamma_n' \in \mathcal{Y}_{\leqslant N}^{+} ([2 d])}}
C_{\Gamma_1, \dots, \Gamma_n}^{(q, N)} (\sigma) 
A_{\Gamma_{\sigma (1)}, \Gamma_1'} 
\dots 
A_{\Gamma_{\sigma (n)}, \Gamma_n'} \otimes 
\wh{\hbar}_{\Gamma_{1}'}^{(N)}
\dots 
\wh{\hbar}_{\Gamma_{n}'}^{(N)}, 
\end{multline}
where $\Gamma_1, \Gamma_2, \dots, \Gamma_n \in \mathcal{Y}_{\leqslant N}^{+} ([2 d])$, 
$n \in \mathbb{Z}_{> 0}$. 
This is a condition on a collection of coefficients $C_{\Gamma_1, \dots, \Gamma_n}^{(q, N)} (\sigma)$.  
Similarly, the \emph{right affine coequivariance condition} is an equality in 
$(\mathcal{E}_{2 d}^{q} (\hbar))_{\leqslant N} \otimes \mathcal{M}_{2 d}^{q}$: 
\begin{multline} 
\label{eq:right_affine_coequivariance} 
\sum_{\substack{\Gamma_1', 
\dots, \Gamma_n' \in \mathcal{Y}_{\leqslant N}^{+} ([2 d]),\\ \sigma \in S_{n}}}
C_{\Gamma_1', 
\dots, \Gamma_n'}^{(q, N)} (\sigma) 
\wh{\hbar}_{\Gamma_{\sigma (1)}'}^{(N)}
\dots 
\wh{\hbar}_{\Gamma_{\sigma (n)}'}^{(N)}
\otimes 
A_{\Gamma_1', \Gamma_1} 
\dots 
A_{\Gamma_n', \Gamma_n}
= \\ = 
\sum_{\substack{\sigma \in S_{n}, \\ 
\Gamma_1', \dots, \Gamma_n' \in \mathcal{Y}_{\leqslant N}^{+} ([2 d])}}
C_{\Gamma_1, \dots, \Gamma_n}^{(q, N)} (\sigma) 
\wh{\hbar}_{\Gamma_{1}'}^{(N)}
\dots 
\wh{\hbar}_{\Gamma_{n}'}^{(N)}  
\otimes 
A_{\Gamma_1', \Gamma_{\sigma (1)}} 
\dots 
A_{, \Gamma_n', \Gamma_{\sigma (n)}},  
\end{multline}
where $\Gamma_1, \Gamma_2, \dots, \Gamma_n \in \mathcal{Y}_{\leqslant N}^{+} ([2 d])$, 
$n \in \mathbb{Z}_{> 0}$. 
This is another condition on the coefficients $C_{\Gamma_1, \dots, \Gamma_n}^{(q, N)} (\sigma)$.  
These coefficients are supposed to define a linear map 
$W_{N}^{(q)}: (\mathcal{A}_{2 d}^{q} (\hbar))_{\leqslant N} \to (\mathcal{E}_{2 d}^{q} (\hbar))_{\leqslant N}$, 
\begin{equation} 
\label{eq:qWeyl_CqN}
W_{N}^{(q)}: \hbar_{\Gamma_1}^{(N)} \hbar_{\Gamma_2}^{(N)} \dots \hbar_{\Gamma_{n}}^{(N)} \mapsto 
\sum_{\sigma \in S_n} C_{\Gamma_1, \Gamma_2, \dots, \Gamma_n}^{(q)} (\sigma)
\wh{\hbar}_{\Gamma_{\sigma (1)}}^{(N)} 
\wh{\hbar}_{\Gamma_{\sigma (2)}}^{(N)} \dots \wh{\hbar}_{\Gamma_{\sigma (n)}}^{(N)}, 
\end{equation}
and to satisfy 
\begin{equation} 
\label{eq:classical_limit}
\sum_{\sigma \in S_n} C_{\Gamma_1, \Gamma_2, \dots, \Gamma_n}^{(q, N)} (\sigma) q (\wt{\Gamma}, \sigma) = 1, 
\end{equation}
where $\wt{\Gamma} = (\Gamma_1, \Gamma_2, \dots, \Gamma_n) \in (\mathcal{Y}_{\leqslant N}^{+} ([2 d]))^{n}$. 
Let us term the latter condition \emph{the existence of a classical limit} for the quantization map. 
Observe, that the notation $q (\wt{\Gamma}, \sigma)$ defined 
in \eqref{eq:q_Gamma_tilde_sigma} can explicitly be described as 
\begin{equation*} 
q (\wt{\Gamma}, \sigma) := \prod_{\substack{
1 \leqslant i < j \leqslant n, \\
\sigma^{-1} (i) > \sigma^{-1} (j)
}} 
q_{\Gamma_j, \Gamma_i}, 
\end{equation*}
where $q_{\Gamma, \Gamma'}$ corresponds to $\hbar_{\Gamma} \hbar_{\Gamma'} = q_{\Gamma', \Gamma} \hbar_{\Gamma'} \hbar_{\Gamma}$, 
for $\Gamma, \Gamma' \in \mathcal{Y} ([2 d])$.

\begin{theorem} 
The conditions of left and right affine 
coequivariance \eqref{eq:left_affine_coequivariance}, 
\eqref{eq:right_affine_coequivariance}, 
in addition to the condition of existence of a classical limit 
\eqref{eq:classical_limit}, determine the ideal $\mathcal{I}$ 
in \eqref{eq:algebra_M} and the coefficients $C_{\Gamma_1, \Gamma_2, \dots, \Gamma_N}^{(q, N)} (\sigma)$ for the 
quantization map \eqref{eq:qWeyl_CqN}, 
\begin{equation*} 
C_{\Gamma_1, \Gamma_2, \dots, \Gamma_n}^{(q, N)} (\sigma) = 
\frac{q (\wt{\Gamma}, \sigma)}{\sum_{\varkappa \in S_{n}} (q (\wt{\Gamma}, \varkappa))^{2}}, 
\end{equation*}
where $\wt{\Gamma} = (\Gamma_1, \Gamma_2, \dots, \Gamma_n) \in 
(\mathcal{Y}_{\leqslant N}^{+} ([2 d]))^{n}$, 
$\sigma \in S_n$, $N \in \mathbb{Z}_{> 0}$.  
The ideal $\mathcal{I}$ is generated by the quantum matrices type relations: 
\begin{multline*} 
\Lambda_{\Gamma_2, \Gamma_1'}
\Lambda_{\Gamma_1, \Gamma_2'} + 
q_{\Gamma_1', \Gamma_2'} 
\Lambda_{\Gamma_2, \Gamma_2'} 
\Lambda_{\Gamma_1, \Gamma_1'} 
= \\ =  
q_{\Gamma_1, \Gamma_2} \big\lbrace 
\Lambda_{\Gamma_1, \Gamma_1'} 
\Lambda_{\Gamma_2, \Gamma_2'} 
+ q_{\Gamma_1', \Gamma_2'} 
\Lambda_{\Gamma_1, \Gamma_2'}
\Lambda_{\Gamma_2, \Gamma_1'}
\big\rbrace, 
\end{multline*}
and 
\begin{multline*} 
\Lambda_{\Gamma_1, \Gamma_2'}
\Lambda_{\Gamma_2, \Gamma_1'}
+ q_{\Gamma_1, \Gamma_2} 
\Lambda_{\Gamma_2, \Gamma_2'}
\Lambda_{\Gamma_1, \Gamma_1'}
= \\ = 
q_{\Gamma_1', \Gamma_2'} 
\lbrace 
\Lambda_{\Gamma_1, \Gamma_1'}
\Lambda_{\Gamma_2, \Gamma_2'}
+ q_{\Gamma_1, \Gamma_2} 
\Lambda_{\Gamma_2, \Gamma_1'}
\Lambda_{\Gamma_1, \Gamma_2'}
\rbrace, 
\end{multline*}
where $\Gamma_1$, $\Gamma_2$, $\Gamma_1'$, and $\Gamma_2'$ vary over 
$\mathcal{Y}^{+} ([2 d])$. 
\end{theorem}

\noindent\emph{Proof.} 
Let us first derive another generic fact about the coefficients $C_{\Gamma_1, \Gamma_2, \dots, \Gamma_n}^{(q, N)} (\sigma)$. 
Fix a positive integer $N$. 
Take a pair of permutations $\sigma, \tau \in S_{n}$ and look at 
the coefficient corresponding to $\sigma \circ \tau \in S_n$. 
On one hand, since $\hbar_{\Gamma_{\tau (1)}}^{(N)} \dots \hbar_{\Gamma_{\tau (n)}}^{(N)} = 
q (\wt{\Gamma}, \tau) \hbar_{\Gamma_1}^{(N)} \dots \hbar_{\Gamma_n}^{(N)}$, 
where $\wt{\Gamma} = (\Gamma_1, \dots, \Gamma_n)$, we must have 
\begin{equation*} 
W_{N}^{(q)} (
\hbar_{\Gamma_{\tau (1)}}^{(N)} \dots \hbar_{\Gamma_{\tau (n)}}^{(N)}) 
= 
q (\wt{\Gamma}, \tau) 
\sum_{\sigma \in S_n} 
C_{\Gamma_1, \dots, \Gamma_n}^{(q, N)} (\sigma) 
\wh{\hbar}_{\Gamma_{\sigma (1)}}^{(N)} \dots 
\wh{\hbar}_{\Gamma_{\sigma (n)}}^{(N)} + 
\wh{Z}, 
\end{equation*} 
where $\wh{Z} \in \mathcal{F}^{m + 1} (\mathcal{E}_{2 d}^{q} (\hbar))_{\leqslant N}$, 
$m = |\Gamma_1| + \dots + |\Gamma_n| - n$. 
On the other hand, a straightforward application of the formula for $W_{N}^{(q)}$ yields 
\begin{equation*} 
W_{N}^{(q)} (
\hbar_{\Gamma_{\tau (1)}}^{(N)} \dots \hbar_{\Gamma_{\tau (n)}}^{(N)}) 
= 
\sum_{\sigma \in S_n} 
C_{\Gamma_{\tau (1)}, \dots, \Gamma_{\tau (n)}}^{(q, N)} (\sigma) 
\wh{\hbar}_{\Gamma_{(\tau \circ \sigma) (1)}}^{(N)} \dots 
\wh{\hbar}_{\Gamma_{(\tau \circ \sigma) (n)}}^{(N)}. 
\end{equation*}
Changing the summation index to $\sigma' = \tau \circ \sigma$ and comparing the two expressions, 
one obtains 
\begin{equation} 
\label{eq:C_tau_sigma}
C_{\Gamma_{\tau (1)}, \dots, \Gamma_{\tau (n)}}^{(q, N)} (\tau^{-1} \circ \sigma) = 
q (\wt{\Gamma}, \tau) C_{\Gamma_1, \dots, \Gamma_n}^{(q, N)} (\sigma). 
\end{equation}
Fix now $\wt{\Gamma} = (\Gamma_1, \Gamma_2, \dots, \Gamma_n) \in 
(\mathcal{Y}_{\leqslant N}^{+} ([2 d]))^{n}$,  
and look at the left affine coequivariance condition \eqref{eq:left_affine_coequivariance}. 
Denote the left hand side $\mathbf{L}_{N} (\wt{\Gamma})$ and the 
right hand side $\mathbf{R}_{N} (\wt{\Gamma})$. 
Changing the summation indices in the expression for $\mathbf{L}_{N} (\wt{\Gamma})$ to 
$\Gamma_i'' = \Gamma_{\sigma (i)}'$, $i \in [n]$, one obtains 
\begin{multline*} 
\mathbf{L}_{N} (\wt{\Gamma}) = 
\sum_{\Gamma_1'', \dots, \Gamma_n'' \in \mathcal{Y}_{\leqslant N}^{+} ([2 d])} 
\sum_{\sigma \in S_n}
A_{\Gamma_1, \Gamma_{\sigma^{-1} (1)}''}
\dots 
A_{\Gamma_n, \Gamma_{\sigma^{-1} (n)}''}
\otimes \\ \otimes  
C_{\Gamma_{\sigma^{-1} (1)}'', 
\dots, \Gamma_{\sigma^{-1} (n)}''}^{(q, N)} (\sigma) 
\wh{\hbar}_{\Gamma_{1}''}^{(N)}
\dots 
\wh{\hbar}_{\Gamma_{n}''}^{(N)}. 
\end{multline*} 
Observe now, that if one has an abstract expression $f (\wt{\Gamma})$ 
depending on 
$\wt{\Gamma} = (\Gamma_1, \dots, \Gamma_n) \in (\mathcal{Y}_{\leqslant N}^{+} ([2 d]))^{n}$
and wishes to take a sum over all tuples $\wt{\Gamma}$, 
then one can do it as follows: 
\begin{equation*} 
\sum_{\wt{\Gamma} \in (\mathcal{Y}_{\leqslant N}^{+} ([2 d]))^{n}} 
f (\wt{\Gamma}) = 
\sum_{
\substack{
\Gamma_1, \Gamma_2, \dots, \Gamma_n \in \mathcal{Y}_{\leqslant N}^{+} ([2 d]),\\
\Gamma_1 \succcurlyeq \Gamma_2 \succcurlyeq \dots \succcurlyeq \Gamma_n
}} 
\frac{1}{w_{\Gamma_1, \Gamma_2, \dots, \Gamma_n}} 
\sum_{\rho \in S_n}
f (\wt{\Gamma}_{\rho}),  
\end{equation*}
where $w_{\Gamma_1, \Gamma_2, \dots, \Gamma_n} := \# \lbrace \sigma \in S_n \,|\, \wt{\Gamma}_{\sigma} = \wt{\Gamma} \rbrace$. 
Hence, 
\begin{multline*} 
\mathbf{L}_{N} (\wt{\Gamma}) = 
\sum_{
\substack{
\Gamma_1'', \dots, \Gamma_n'' \in \mathcal{Y}_{\leqslant N}^{+} ([2 d]),\\
\Gamma_1'' \succcurlyeq \dots \succcurlyeq \Gamma_n'',\\
\sigma, \rho \in S_n
}} 
\frac{1}{w_{\Gamma_1'', \dots, \Gamma_n''}} 
A_{\Gamma_1, \Gamma_{(\rho \circ \sigma^{-1}) (1)}''}
\dots 
A_{\Gamma_n, \Gamma_{(\rho \circ \sigma^{-1}) (n)}''}
\otimes \\ \otimes  
C_{\Gamma_{(\rho \circ \sigma^{-1}) (1)}'', 
\dots, \Gamma_{(\rho \circ \sigma^{-1}) (n)}''}^{(q, N)} (\sigma) 
\wh{\hbar}_{\Gamma_{\rho (1)}''}^{(N)}
\dots 
\wh{\hbar}_{\Gamma_{\rho (n)}''}^{(N)}. 
\end{multline*}
The same trick applied to $\mathbf{R}_{N} (\wt{\Gamma})$ yields: 
\begin{multline} 
\label{eq:rhs_left_coequivariance} 
\mathbf{R}_{N} (\wt{\Gamma}) = 
\sum_{\substack{\sigma, \rho  \in S_{n}, \\ 
\Gamma_1', \dots, \Gamma_n' \in \mathcal{Y}_{\leqslant N}^{+} ([2 d]),\\ 
\Gamma_1' \succcurlyeq \dots \succcurlyeq \Gamma_n'
}}
\frac{1}{w_{\Gamma_1', \dots, \Gamma_n'}}
A_{\Gamma_{\sigma (1)}, \Gamma_{\rho(1)}'} 
\dots 
A_{\Gamma_{\sigma (n)}, \Gamma_{\rho(n)}'} 
\otimes \\ \otimes  
C_{\Gamma_1, \dots, \Gamma_n}^{(q, N)} (\sigma) 
\wh{\hbar}_{\Gamma_{\rho(1)}'}^{(N)}
\dots 
\wh{\hbar}_{\Gamma_{\rho(n)}'}^{(N)}. 
\end{multline}
Now, if we look at the 
products of $\hbar_{\Gamma}^{(N)}$ in the associated graded and take the leading term, then  
the equality $\mathbf{L}_{N} (\wt{\Gamma}) = \mathbf{R}_{N} (\wt{\Gamma})$ leaves us with 
\begin{multline*} 
\sum_{\sigma, \rho \in S_n} 
A_{\Gamma_1, \Gamma_{(\rho \circ \sigma^{-1}) (1)}'}
\dots 
A_{\Gamma_n, \Gamma_{(\rho \circ \sigma^{-1}) (n)}'}
C_{\Gamma_{(\rho \circ \sigma^{-1}) (1)}', 
\dots, \Gamma_{(\rho \circ \sigma^{-1}) (n)}'}^{(q, N)} (\sigma) 
\times \\ \times 
q (\wt{\Gamma}', \rho) 
=
\sum_{\sigma, \rho \in S_n}
A_{\Gamma_{\sigma (1)}, \Gamma_{\rho(1)}'} 
\dots 
A_{\Gamma_{\sigma (n)}, \Gamma_{\rho(n)}'} 
C_{\Gamma_1, \dots, \Gamma_n}^{(q, N)} (\sigma) 
q (\wt{\Gamma}', \rho),  
\end{multline*} 
for every $n \in \mathbb{Z}_{> 0}$, every $\wt{\Gamma} = (\Gamma_1, \Gamma_2, \dots, \Gamma_n) \in 
(\mathcal{Y}_{\leqslant N}^{+} ([2 d]))^{n}$, and every 
$\wt{\Gamma}' = (\Gamma_1', \Gamma_2', \dots, \Gamma_n') \in 
(\mathcal{Y}_{\leqslant N}^{+} ([2 d]))^{n}$, such that 
$\Gamma_1' \succcurlyeq \Gamma_2' \succcurlyeq \dots \succcurlyeq \Gamma_n'$. 
The left hand side of this equality can be simplified, 
if one takes into account the property \eqref{eq:q_Gamma_tilde_sigma_tau}, which implies 
\begin{equation*} 
q (\wt{\Gamma}', \rho) = 
q (\wt{\Gamma}', (\rho \circ \sigma^{-1}) \circ \sigma) = 
q (\wt{\Gamma}', \rho \circ \sigma^{-1}) q (\Gamma_{\rho \circ \sigma^{-1}}', \sigma). 
\end{equation*}
Introducing an index of summation 
$\varkappa = \rho \circ \sigma^{-1}$ in place of $\rho$, and 
invoking the condition of existence of classical limit \eqref{eq:classical_limit}, one obtains: 
\begin{multline} 
\label{eq:left_coequivariance_Ln} 
\sum_{\varkappa \in S_{n}} 
A_{\Gamma_1, \Gamma_{\varkappa (1)}'} \dots 
A_{\Gamma_n, \Gamma_{\varkappa (n)}'} 
q (\wt{\Gamma}', \varkappa) 
= \\ = 
\sum_{\sigma, \rho \in S_n}
A_{\Gamma_{\sigma (1)}, \Gamma_{\rho(1)}'} 
\dots 
A_{\Gamma_{\sigma (n)}, \Gamma_{\rho(n)}'} 
C_{\Gamma_1, \dots, \Gamma_n}^{(q, N)} (\sigma) 
q (\wt{\Gamma}', \rho),  
\end{multline} 
Consider now the case $n = 2$. 
For every $\Gamma_1, \Gamma_2, \Gamma_1', \Gamma_2' \in \mathcal{Y}_{\leqslant N}^{+} ([2 d])$, 
such that $\Gamma_1' \succcurlyeq \Gamma_2'$, 
since $q (\wt{\Gamma}', \mathit{id}) = 1$, we have 
\begin{multline*} 
A_{\Gamma_1, \Gamma_1'} A_{\Gamma_2, \Gamma_2'} 
+ A_{\Gamma_1, \Gamma_2'} A_{\Gamma_2, \Gamma_1'} q (\wt{\Gamma}', (12)) 
= \\ = 
C_{\Gamma_1, \Gamma_2}^{(q, N)} (\mathit{id}) \big\lbrace 
A_{\Gamma_1, \Gamma_1'} A_{\Gamma_2, \Gamma_2'} 
+ A_{\Gamma_1, \Gamma_2'} A_{\Gamma_2, \Gamma_1'} q (\wt{\Gamma}', (12))
\big\rbrace 
+ \\ + 
C_{\Gamma_1, \Gamma_2}^{(q, N)} ((12)) \big[ 
A_{\Gamma_2, \Gamma_1'} A_{\Gamma_1, \Gamma_2'} 
+ A_{\Gamma_2, \Gamma_2'} A_{\Gamma_1, \Gamma_1'} 
q (\wt{\Gamma}', (12))
\big], 
\end{multline*} 
where $\mathit{id}, (12) \in S_2$ are the two elements of the symmetric group $S_2$. 
From the classical limit condition \eqref{eq:classical_limit}, one obtains: 
\begin{equation*} 
C_{\Gamma_1, \Gamma_2}^{(q, N)} (\mathit{id}) +  
C_{\Gamma_1, \Gamma_2}^{(q, N)} ((12)) q (\wt{\Gamma}, (12)) = 1. 
\end{equation*}
Observe that $q (\wt{\Gamma}, (12)) = q_{\Gamma_1, \Gamma_2}$. 
Express now $C_{\Gamma_1, \Gamma_2}^{(q, N)} (\mathit{id})$ via $C_{\Gamma_1, \Gamma_2}^{(q, N)} ((12))$ 
and substitute the result into the previous equality. Cancelling out $C_{\Gamma_1, \Gamma_2}^{(q, N)} ((12))$, 
one arrives at 
\begin{multline} 
\label{eq:quantum_matrices_left}
\big[ 
A_{\Gamma_2, \Gamma_1'} A_{\Gamma_1, \Gamma_2'} 
+ A_{\Gamma_2, \Gamma_2'} A_{\Gamma_1, \Gamma_1'} 
q_{\Gamma_1', \Gamma_2'}
\big] 
- \\ -
q_{\Gamma_1, \Gamma_2} 
\big\lbrace 
A_{\Gamma_1, \Gamma_1'} A_{\Gamma_2, \Gamma_2'} 
+ A_{\Gamma_1, \Gamma_2'} A_{\Gamma_2, \Gamma_1'} q_{\Gamma_1', \Gamma_2'}
\big\rbrace = 0. 
\end{multline}
Since the level of truncation $N$ is arbitrary, this corresponds precisely to the 
first half 
of the defining relations of $\mathcal{M}_{2 d}^{(q)}$. 
The second half 
stems in a totally similar manner from the right coequivariance condition 
\eqref{eq:right_affine_coequivariance}.  

We still need to define the coefficients $C_{\Gamma_1, \dots, \Gamma_n}^{(q, N)} (\sigma)$. Let $n \in \mathbb{Z}_{> 0}$. 
Denote 
\begin{equation*} 
L_n (\wt{\Gamma}, \wt{\Gamma}') := \sum_{\varkappa \in S_{n}} 
A_{\Gamma_1, \Gamma_{\varkappa (1)}'} \dots 
A_{\Gamma_n, \Gamma_{\varkappa (n)}'} 
q (\wt{\Gamma}', \varkappa), 
\end{equation*}
where $\wt{\Gamma} = (\Gamma_1, \dots, \Gamma_n)$ and  
$\wt{\Gamma}' = (\Gamma_1', \dots, \Gamma_n')$ are in $(\mathcal{Y}_{\leqslant N}^{+} ([2 d]))^{n}$. 
Note, that $L_n (\wt{\Gamma}, \wt{\Gamma}')$ is precisely the left-hand side of 
\eqref{eq:left_coequivariance_Ln} expressing the 
left coequivariance condition \eqref{eq:left_affine_coequivariance}, so there is also a right analogue 
\begin{equation*} 
R_n (\wt{\Gamma}, \wt{\Gamma}') := \sum_{\varkappa \in S_{n}} 
A_{\Gamma_{\varkappa (1)}, \Gamma_{1}'} \dots 
A_{\Gamma_{\varkappa (n)}, \Gamma_{n}'} 
q (\wt{\Gamma}, \varkappa), 
\end{equation*}
which corresponds to the right coequivariance condition \eqref{eq:right_affine_coequivariance}. 
We claim that the following properties hold: 
\begin{equation} 
\label{eq:Ln_Rn_properties}
L_n (\wt{\Gamma}_{\sigma}, \wt{\Gamma}') = q (\wt{\Gamma}, \sigma) L_n (\wt{\Gamma}, \wt{\Gamma}'), \quad 
R_n (\wt{\Gamma}, \wt{\Gamma}_{\sigma}') = q (\wt{\Gamma}', \sigma) R_n (\wt{\Gamma}, \wt{\Gamma}'), 
\end{equation}
for any $\sigma \in S_n$. This can be done by induction in $n$. 
Consider $L_n (\wt{\Gamma}, \wt{\Gamma}')$, for example. 
If $n = 2$, then one arrives at \eqref{eq:quantum_matrices_left}. 
If $n > 2$, then proceed as follows. 
Observe, that if the property mentioned holds 
for some given $\sigma \in S_n$ and $\tau \in S_n$, and for any $\wt{\Gamma}$ and $\wt{\Gamma}'$, 
then 
\begin{equation*} 
L_n (\wt{\Gamma}_{\sigma \circ \tau}, \wt{\Gamma}') = 
q (\wt{\Gamma}_{\sigma}, \tau) L_n (\wt{\Gamma}_{\sigma}, \wt{\Gamma}') = 
q (\wt{\Gamma}_{\sigma}, \tau) q (\wt{\Gamma}, \sigma) L_n (\wt{\Gamma}, \wt{\Gamma}'), 
\end{equation*}
and it holds for $\sigma \circ \tau \in S_n$ as well due to \eqref{eq:q_Gamma_tilde_sigma_tau}, 
Therefore, it suffices to check it only on the generators of $S_n$. 
Rewrite $L_n (\wt{\Gamma}, \wt{\Gamma}')$ as follows: 
\begin{multline*} 
L_n (\wt{\Gamma}, \wt{\Gamma}') = 
\sum_{m = 1}^{n} \sum_{\rho \in S_{n - 1}} 
A_{\Gamma_1, \Gamma_{((mn) \circ \rho_{n}^{+}) (1)}'} \dots 
A_{\Gamma_1, \Gamma_{((mn) \circ \wt{\rho}) (n - 1)}'} 
A_{\Gamma_n, \Gamma_m'} 
\times \\ \times
q (\wt{\Gamma}', (mn)) q (\wt{\Gamma}_{(mn)}')
= \sum_{m = 1}^{n} L_{n - 1} (\wt{\Gamma}_{\leqslant n - 1}, (\wt{\Gamma}_{(mn)}')_{\leqslant n - 1}) 
q (\wt{\Gamma}', (mn)), 
\end{multline*}
where $\rho_{n}^{+}$ denotes the canonical image of $\rho \in S_{n - 1}$ in $S_n$, such that 
$\rho_{n}^{+} (n) = n$ and $\rho_{n}^{+} (i) = \rho (i)$, $i < n$, 
and the symbol  
$(mn) \in S_n$ denotes the transposition of $m$ and $n$, 
and $(-)_{\leqslant n - 1}$ corresponds to a truncation of a string of symbols, so that  
$\wt{\Gamma}_{\leqslant n - 1} = (\Gamma_1, \dots, \Gamma_{n - 1})$, 
and, similarly, for $\wt{\Gamma}_{(mn)}'$. 
If one now makes an inductive assumption 
that \eqref{eq:Ln_Rn_properties} holds for $n - 1$, 
then this implies $L_n (\wt{\Gamma}_{\lambda_{n}^{+}}, \wt{\Gamma}') = 
q (\wt{\Gamma}, \lambda_{n}^{+}) L_n (\wt{\Gamma}, \wt{\Gamma}')$, for all $\lambda \in S_{n - 1}$. 
In a totally similar way, 
isolating the first, but not the $n$-th factor in the products, 
one can show, that 
$L_n (\wt{\Gamma}_{\mu_{1}^{+}}, \wt{\Gamma}') = 
q (\wt{\Gamma}, \lambda_{1}^{+}) L_n (\wt{\Gamma}, \wt{\Gamma}')$, 
for any $\mu \in S_{n -1}$, where $\mu_{1}^{+} \in S_n$ is the permutation, such that 
$\mu_{1}^{+} (1) = 1$, and $\mu_{1}^{+} (i + 1) = \mu (i)$, $i \in [n - 1]$.  
Since the collection of permutations of the shape $\lambda_{n}^{+}$ and $\mu_{1}^{+}$, 
where $\lambda, \mu \in S_{n - 1}$, generate the whole $S_n$, the property claimed follows. 
The second equality in \eqref{eq:Ln_Rn_properties} is established in a similar way. 

Return now to the left coequivariance condition $\mathbf{L}_{N} (\wt{\Gamma}) = \mathbf{R}_{N} (\wt{\Gamma})$, 
where $\wt{\Gamma} = (\Gamma_1, \dots, \Gamma_n) \in (\mathcal{Y}_{\leqslant N}^{+} ([2 d]))^n$. 
Consider first the case where all $\Gamma_i = \Gamma_0 \in \mathcal{Y}_{\leqslant N}^{+} ([2 d])$, $i \in [n]$. 
The commutation relations for the ideal $\mathcal{I}$ (more precisely, those that stem from the \emph{right} 
coequivariance condition), imply 
an equality 
$A_{\Gamma_0, \Gamma_{\lambda (1)}'} \dots A_{\Gamma_0, \Gamma_{\lambda (n)}'} = q (\wt{\Gamma}', \lambda) 
A_{\Gamma_0, \Gamma_1'} \dots A_{\Gamma_0, \Gamma_n'}$, for any $\lambda \in S_n$. 
Therefore, the 
requirement 
$\mathbf{L}_{N} (\wt{\Gamma}) = \mathbf{R}_{N} (\wt{\Gamma})$ 
acquires 
in this case 
the shape: 
\begin{multline*} 
\sum_{\substack{
\Gamma_1', \dots, \Gamma_n' \in \mathcal{Y}_{\leqslant N}^{+} ([2 d]),\\
\Gamma_1' \succcurlyeq \dots \succcurlyeq \Gamma_n'
}}
\sum_{\sigma, \rho \in S_n} 
A_{\Gamma_0, \Gamma_1'} A_{\Gamma_0, \Gamma_2'} \dots A_{\Gamma_0, \Gamma_n'} \otimes 
\wh{\hbar}_{\Gamma_{\rho (1)}'} \wh{\hbar}_{\Gamma_{\rho (2)}'} \dots \wh{\hbar}_{\Gamma_{\rho (n)}'} 
\times \\ \times 
\Big[ 
q (\wt{\Gamma}', \rho \circ \sigma^{-1}) 
C_{\Gamma_{(\rho \circ \sigma^{-1}) (1)}', \dots, \Gamma_{(\rho \circ \sigma^{-1}) (n)}'}^{(q, N)} (\sigma)
- q (\wt{\Gamma}', \rho) 
C_{\Gamma_0, \dots, \Gamma_0}^{(q, N)} (\sigma)
\Big] = 0, 
\end{multline*}
where $\wt{\Gamma}' := (\Gamma_1', \Gamma_2', \dots, \Gamma_n')$. 
The 
existence of the 
classical limit condition and the fact $q (\wt{\Gamma}^{(0)}, \sigma) = 1$ imply 
$\sum_{\sigma \in S_n} C_{\Gamma_0, \dots, \Gamma_0}^{(q, N)} (\sigma) = 1$. 
Invoking the property 
\eqref{eq:C_tau_sigma},  
one obtains 
$C_{\Gamma_{(\rho \circ \sigma^{-1}) (1)}', \dots, \Gamma_{(\rho \circ \sigma^{-1}) (n)}'}^{(q, N)} (\sigma) = 
q (\wt{\Gamma}', \rho \circ \sigma^{-1}) C_{\Gamma' (1), \dots, \Gamma' (n)}^{(q, N)} (\rho)$. 
Performing the summation over $\varkappa = \rho \circ \sigma^{-1} \in S_n$ yields: 
\begin{multline*} 
\sum_{\substack{
\Gamma_1', \dots, \Gamma_n' \in \mathcal{Y}_{\leqslant N}^{+} ([2 d]),\\
\Gamma_1' \succcurlyeq \dots \succcurlyeq \Gamma_n'
}}
\sum_{\sigma, \rho \in S_n} 
A_{\Gamma_0, \Gamma_1'} A_{\Gamma_0, \Gamma_2'} \dots A_{\Gamma_0, \Gamma_n'} \otimes 
\wh{\hbar}_{\Gamma_{\rho (1)}'} \wh{\hbar}_{\Gamma_{\rho (2)}'} \dots \wh{\hbar}_{\Gamma_{\rho (n)}'} 
\times \\ \times 
\Big[ 
Z_n^{(q)} (\wt{\Gamma}')
C_{\Gamma_{1}', \dots, \Gamma_{n}'}^{(q, N)} (\sigma)
- q (\wt{\Gamma}', \rho) 
\Big] = 0, 
\end{multline*}
where 
\begin{equation*}
Z_n^{(q)} (\wt{\Gamma}) := \sum_{\sigma \in S_n} (q (\wt{\Gamma}, \sigma))^{2}, 
\end{equation*}
for any $\wt{\Gamma} = (\Gamma_1, \dots, \Gamma_n) \in (\mathcal{Y} ([2 d]))^n$. 
If we look at this expression in the associated graded, 
i.e. the product $\wh{\hbar}_{\Gamma_{\rho (1)}'} \dots \wh{\hbar}_{\Gamma_{\rho (n)}'}$ 
becomes $\hbar_{\Gamma_{\rho (1)}'} \dots \hbar_{\Gamma_{\rho (n)}'} = 
q (\wt{\Gamma}', \rho) \hbar_{\Gamma_1'} \dots \hbar_{\Gamma_n'}$, then it follows, that 
the only candidate for the coefficients in the case $\Gamma_1' \succcurlyeq \dots \succcurlyeq \Gamma_n'$ is just  
$C_{\Gamma_{1}', \dots, \Gamma_{n}'}^{(q, N)} (\sigma) = q (\wt{\Gamma}', \rho)/ Z_{n}^{(q)} (\wt{\Gamma}') 
=: \bar C_{\wt{\Gamma}'}^{(q)} (\sigma)$.  
If we use this expression for \emph{all} $\wt{\Gamma}'$, not just $\Gamma_1' \succcurlyeq \dots \succcurlyeq \Gamma_n'$, 
then the property $q (\wt{\Gamma}', \rho \circ \sigma) = q (\wt{\Gamma}', \rho) q (\wt{\Gamma}_{\rho}', (\sigma))$ 
implies that $\bar C_{\wt{\Gamma}_{\rho}'}^{(q)} (\sigma) = q (\wt{\Gamma}', \rho) 
\bar C_{\wt{\Gamma}'}^{(q)} (\rho \circ \sigma)$, 
just what is needed to satisfy \eqref{eq:C_tau_sigma}. 
It remains to check if the equation 
$\mathbf{L}_{N} (\wt{\Gamma}) = \mathbf{R}_{N} (\wt{\Gamma})$, 
where $\wt{\Gamma} = (\Gamma_1, \dots, \Gamma_n) \in (\mathcal{Y}_{\leqslant N}^{+} ([2 d]))^n$, 
is indeed satisfied. 
Using the property of the coefficients $\bar C_{\wt{\Gamma}'}^{(q)} (\sigma)$ mentioned, one obtains: 
\begin{equation*} 
\mathbf{L}_{N} (\wt{\Gamma}) = 
\sum_{\substack{\Gamma_1' \succcurlyeq \dots \succcurlyeq \Gamma_n,\\
\rho \in S_n}} \frac{1}{w_{\wt{\Gamma}'}}
A_{\Gamma_1, \Gamma_{\rho (1)}'} \dots 
A_{\Gamma_n \Gamma_{\rho (n)}'} 
\bar C_{\wt{\Gamma}'}^{(q)} (\rho) \otimes
H_{\wt{\Gamma}}^{(q)}, 
\end{equation*}
where $H_{\wt{\Gamma}'}^{(q)} := 
\sum_{\varkappa \in S_n}
q (\wt{\Gamma}', \varkappa) 
\wh{\hbar}_{\Gamma_{\varkappa (1)}'} \dots \wh{\hbar}_{\Gamma_{\varkappa (n)}'}$. 
Now, using a trick of ``inserting a unit'' $\sum_{\sigma \in S_n} \bar C_{\wt{\Gamma}}^{(q)} (\sigma) q (\wt{\Gamma}, \sigma) = 1$, 
and then invoking the property $L_{n} (\wt{\Gamma}_{\sigma}, \wt{\Gamma}') = q (\wt{\Gamma}, \sigma) L_n (\wt{\Gamma}, \wt{\Gamma}')$, 
one arrives at 
\begin{equation}
\label{eq:lhs_left_coequivariance_reduced} 
\mathbf{L}_{N} (\wt{\Gamma}) = 
\sum_{
\substack{
\wt{\Gamma}' \in (\mathcal{Y}_{\leqslant N}^{+} ([2 d]))^{n},\\ 
\Gamma_1' \succcurlyeq \dots \succcurlyeq \Gamma_n'}} \frac{1}{w_{\wt{\Gamma}'}}
A_{\wt{\Gamma}, \wt{\Gamma}'}^{(q)} 
\otimes 
H_{\wt{\Gamma}'}^{(q)}, 
\end{equation} 
where $A_{\wt{\Gamma}, \wt{\Gamma}'}^{(q)} := \sum_{\sigma, \rho \in S_n} 
\bar C_{\wt{\Gamma}}^{(q)} (\sigma) 
A_{\Gamma_1, \Gamma_{\rho (1)}'} \dots 
A_{\Gamma_n \Gamma_{\rho (n)}'} 
\bar C_{\wt{\Gamma}'}^{(q)} (\rho)$. 
For the right-hand side $\mathbf{R}_{N} (\wt{\Gamma})$ in the shape \eqref{eq:rhs_left_coequivariance}, 
one can first invoke the property $R_n (\wt{\Gamma}, \wt{\Gamma}_{\rho}') = 
q (\wt{\Gamma}', \rho) R_n (\wt{\Gamma}, \wt{\Gamma}')$ to obtain 
\begin{equation*}
\mathbf{R}_{N} (\wt{\Gamma}) = 
\sum_{\substack{\Gamma_1' \succcurlyeq \dots \succcurlyeq \Gamma_n',\\
\sigma \in S_{n}}} 
\frac{1}{w_{\wt{\Gamma}'}} 
\bar C_{\wt{\Gamma}}^{(q)} (\sigma) 
A_{\Gamma_{\sigma (1)}, \Gamma_1'} \dots 
A_{\Gamma_{\sigma (n)}, \Gamma_n'} 
\otimes H_{\wt{\Gamma}'}^{(q)}.
\end{equation*}
Now, inserting the unit $\sum_{\varkappa \in S_n} \bar C_{\wt{\Gamma}'}^{(q)} (\varkappa) q (\wt{\Gamma}', \varkappa) = 1$, 
and using again the property $R_n (\wt{\Gamma}, \wt{\Gamma}_{\varkappa}') = 
q (\wt{\Gamma}', \varkappa) R_n (\wt{\Gamma}, \wt{\Gamma}')$, 
one obtains the same expression \eqref{eq:lhs_left_coequivariance_reduced} as for $\mathbf{L}_N (\wt{\Gamma})$. 
Therefore, the left coequivariance requirement \eqref{eq:left_affine_coequivariance} is satisfied, 
$\mathbf{L}_N (\wt{\Gamma}) = \mathbf{R}_N (\wt{\Gamma})$. 
The right coequivariance is established in a totally similar way. 
\qed

It is quite remarkable, that the coefficients $C_{\Gamma_1, \dots, \Gamma_N}^{(q, N)} (\sigma)$ 
for the $q$-Weyl quantization 
$W_{N}^{(q)}: (\mathcal{A}_{2 d}^{q} (\hbar))_{\leqslant N} \to (\mathcal{E}_{2 d}^{q} (\hbar))_{\leqslant N}$
that emerge in the proof do \emph{not} depend on the level of truncation $N$. 
Therefore, we immediately obtain the projective limit 
$W^{(q)}: \mathcal{A}_{2 d}^{q} (\hbar) \to \mathcal{E}_{2 d}^{q} (\hbar)$ of $W_{N}^{(q)}$
as $N \to \infty$, 
\begin{equation} 
\label{eq:qWeyl_quantization_map}
W^{(q)}: \hbar_{\Gamma_1} \hbar_{\Gamma_2} \dots \hbar_{\Gamma_n} \mapsto 
\frac{1}{Z_{n}^{(q)} (\wt{\Gamma})} \sum_{\sigma \in S_n} 
q (\wt{\Gamma}, \sigma) 
\wh{\hbar}_{\Gamma_{\sigma (1)}} \wh{\hbar}_{\Gamma_{\sigma (2)}} \dots \wh{\hbar}_{\Gamma_{\sigma (n)}}, 
\end{equation}
for $\wt{\Gamma} = (\Gamma_1, \Gamma_2, \dots, \Gamma_n) \in \mathcal{Y}^{+} ([2 d])$, 
where the coefficient 
$q (\wt{\Gamma}, \sigma)$ is determined by 
$\hbar_{\Gamma (\sigma (1))} \hbar_{\Gamma (\sigma (2))}\dots \hbar_{\Gamma (\sigma (n))} = 
q (\wt{\Gamma}, \sigma) \hbar_{\Gamma_1} \hbar_{\Gamma_2} \dots \hbar_{\Gamma_n}$, 
and the ``partition function'' $Z_{n}^{(q)} (\wt{\Gamma})$ is just 
$Z_{n}^{(q)} (\wt{\Gamma}) = \sum_{\sigma \in S_n} (q (\wt{\Gamma}, \sigma) )^{2}$. 

\begin{definition} 
The map $W^{(q)}: \mathcal{A}_{2 d}^{q} (\hbar) \to \mathcal{E}_{2 d}^{q} (\hbar)$ 
defined by \eqref{eq:qWeyl_quantization_map} is termed \emph{the $q$-Weyl quantization map}, 
$q = \| q_{i, j} \|$, $i, j \in [2 d]$.  
\end{definition}

The map $W^{(q)}$ allows to induce another star product $\circledast$ on the classical shadow $\mathcal{A}_{2 d}^{q} (\hbar)$, 
\begin{equation*} 
W^{(q)} (f) W^{(q)} (g) = W^{(q)} (f \circledast g), 
\end{equation*}
for $f, g \in \mathcal{A}_{2 d}^{q} (\hbar)$. 
Recall, that we already have a product $\star$ on $\mathcal{A}_{2 d}^{q} (\hbar)$, stemming from the normal 
quantization, and an algebra isomorphism $\varphi: (\mathcal{A}_{2 d}^{q} (\hbar), \star) \overset{\sim}{\to} 
\mathcal{E}_{2 d}^{q} (\hbar)$. Denote $\pi_{m}: \mathcal{A}_{2 d}^{q} (\hbar) \to 
\mathcal{F}^{m} \mathcal{E}_{2 d}^{q} (\hbar)/ \mathcal{F}^{m + 1} \mathcal{E}_{2 d}^{q} (\hbar)$ 
the canonical projection in the $m$-th component of the associated graded, $m \in \mathbb{Z}_{\geqslant 0}$, and 
perceive $\varphi$ as a vector space map. 
Since for any homogeneous element $f \in \mathcal{A}_{2 d}^{q} (\hbar)$ 
of degree $|f|$ we have 
$W (f) - \varphi (f) \in \mathcal{F}^{|f| + 1} \mathcal{E}_{2 d}^{q} (\hbar)$, 
it follows that 
\begin{equation*} 
(\pi_{m} \circ \varphi^{-1}) 
\Big[ 
W^{(q)} (f) W^{(q)} (g) - 
W^{(q)} \Big( \sum_{l = 0}^{m -1} \pi_{l} (f \circledast g) \Big) 
\Big] =  
\pi_m (f \circledast g), 
\end{equation*}
where $m \in \mathbb{Z}_{\geqslant 0}$. 
Applying this formula recursively, one obtains an explicit expression for 
every component of $f \circledast g$.  
\begin{proposition} 
For every $f, g \in \mathcal{A}_{2 d}^{q} (\hbar)$ and every $m \in \mathbb{Z}_{\geqslant 0}$, 
the following holds: 
\begin{multline*} 
\pi_m (f \circledast g) = 
(\pi_m \circ \varphi^{-1}) \sum_{r = 0}^{m} (-1)^r
\sum_{0 \leqslant l_1 < l_2 < \dots < l_r < m}
(W^{(q)} \circ \pi_{l_r} \circ \varphi^{-1}) \circ \dots \\
\dots 
\circ (W^{(q)} \circ \pi_{l_2} \circ \varphi^{-1}) 
\circ (W^{(q)} \circ \pi_{l_1} \circ \varphi^{-1})
\big[ W^{(q)} (f) W^{(q)} (g) \big]. 
\end{multline*}
\end{proposition}

\begin{proof}
Induction by $m = 0, 1, 2, \dots$. 
\qed
\end{proof}

Let us extract the ``first semiclassical correction'' from the product $\circledast$ 
corresponding to the $q$-Weyl quantization $W^{(q)}$, 
\begin{equation*} 
\langle f, g \rangle^{(q)} := \pi_{|f| + |g| + 1} (f \circledast g), 
\end{equation*} 
where $f, g \in \mathcal{A}_{2 d}^{q} (\hbar)$ are homogeneous elements of degrees $|f|$ and $|g|$, respectively. 

\begin{definition} 
The linear graded map $\langle -, - \rangle^{(q)}: \mathcal{A}_{2 d}^{q} (\hbar) \otimes 
\mathcal{A}_{2 d}^{q} (\hbar) \to \mathcal{A}_{2 d}^{q} (\hbar)$ of degree $+ 1$ is termed 
\emph{the canonical distortion} of the $q$-Poisson bracket on the classical shadow $\mathcal{A}_{2 d}^{q} (\hbar)$. 
\end{definition}

Look at the $q$-antisymmetrization $\langle f, g \rangle_{-}^{(q)} := 
\langle f, g \rangle^{(q)} - q (g, f) \langle g, f \rangle^{(q)}$, 
where $q (g, f)$ is determined by $f g = q (g, f) g f$, $f$ and $g$ are homogeneous. 
Extending this bilinearly, one obtains another bracket on $\mathcal{A}_{2 d}^{q} (\hbar)$. 
In quantum mechanics of a $d$-dimensional system with coordinates $x = (x_1, x_2, \dots, x_d)$ and 
the canonically conjugate momenta $p = (p_1, p_2, \dots, p_d)$, the first semiclassical correction 
to the Weyl product of a pair of classical observables $F (x, p)$ and $G (x, p)$ is of the shape 
$(- \ii \hbar/ 2) \lbrace F, G \rbrace$, where $\lbrace -, - \rbrace$ is the canonical Poisson bracket, 
while the antisymmetrization corresponds to 
$- \ii \hbar \lbrace F, G \rbrace$. 
Up to a factor $\lambda = 1/ 2$ independent on $F$ and $G$, these two results coincide. 
In the $q$-deformed case the two brackets are different, 
i.e. such $\lambda$ (possibly, depending on $q$) does not exist. 

\begin{proposition} 
The canonical distortion $\langle -, - \rangle^{(q)}$ of the $q$-Poisson bracket 
satisfies the 2-cocycle condition 
\begin{equation*} 
u \langle v, w \rangle^{(q)} - 
\langle u v, w \rangle^{(q)} + 
\langle u, v w \rangle^{(q)} - 
\langle u, v \rangle^{(q)} w = 0, 
\end{equation*}
for any $u, v, w \in \mathcal{A}_{2 d}^{q} (\hbar)$. 
The $q$-antisymmetrized bracket $\langle -, - \rangle_{-}^{(q)}$ 
yields a $q$-Poisson structure on $\mathcal{A}_{2 d}^{q} (\hbar)$. 
\end{proposition}

\begin{proof}
These facts are a straightforward consequence of the associativity of the $q$-Weyl product $\circledast$
on the $q$-affine space $\mathcal{A}_{2 d}^{q} (\hbar)$.  
\qed 
\end{proof}

It is important to point out that 
the bracket $\langle -, -\rangle^{(q)}$ stemming from the  
the $q$-Weyl quantization is \emph{not} bilinear with respect to 
the generators $\hbar_{\Gamma}$ of degree $|\Gamma| \geqslant 2$. 
In this sense, these generators 
``distort'' already the \emph{classical} picture (i.e. the Poisson bracket)
and should be perceived as dynamical variables just as $\hbar_{\Gamma}$ with $|\Gamma| = 1$, 
which correspond to the classical coordinates $x = (x_1, x_2, \dots, x_d)$ and momenta $p = (p_1, p_2, \dots, p_d)$.

\section{Noncanonical distortions}

In this section we would like to discuss briefly some other possibilities to distort the 
classical Poisson bracket $\lbrace -, -\rbrace$. Informally, the basic idea of \emph{``distortion''} 
is to add \emph{new} variables into the picture and to extend the bracket in a nontrivial way. 
We have introduced the epoch\'e algebra $\mathcal{E}_{2 d}^{q} (\hbar)$ as a kind of \emph{noncommutative neighbourhood} 
of the semiclassical parameter $\hbar$ of a mechanical system with $d$ degrees of freedom, 
but, of course, for the mathematical construction the nature of this parameter is not important. 
One encounters the canonical commutation and anticommutation relations, for example, in statistical physics 
considering the creation and annihilation operators of the particles constituting a system (different types of 
bosons and fermions). Extending these relations with a central parameter $g$ (the interaction parameter), 
the generic shape of these relations is as follows: 
\begin{equation*} 
\begin{gathered}
\psi_{i}^{-} \psi_{j}^{+} - \varepsilon_{j, i} \psi_{j}^{+} \psi_{i}^{-} = g \delta_{i, j}, \\
\psi_{i}^{-} \psi_{j}^{-}  - \varepsilon_{j, i} \psi_{j}^{-} \psi_{i}^{-} = 0, \quad 
 \psi_{i}^{+} \psi_{j}^{+}  - \varepsilon_{j, i} \psi_{j}^{+} \psi_{i}^{+} = 0
\end{gathered}
\end{equation*}
where $i, j \in \mathbb{Z}_{\geqslant 0}$, $\psi_{i}^{+}$ are the creation operators, and 
$\psi_{i}^{-}$ are the annihilation operators, and $\varepsilon_{i, j} \in \lbrace -1, +1 \rbrace$ 
depending on the statistics of a pair of particles associated with the indices $i$ and $j$. 
The difference is that the number of degrees of freedom $d$ is now infinite, but 
one can still introduce the analogues of coordinates $x_i$ and momenta $p_i$ as 
$\psi_i^{\pm} = (\wh{x}_i \mp \ii \wh{p}_i)/ \sqrt{2}$. 
In the condensed matter physics, it is a common practice to consider an asymptotics with respect to $g \to 0$ 
(the \emph{quasiparticle approximation}), and one may consider a noncommutative neighbourhood 
around $g$ (this parameter is sometimes termed the \emph{external} Planck constant).  
Look at a second quantized observable $\wh{B}$ of polynomial type: 
\begin{equation*} 
\wh{B} = \sum_{m = 1}^{\infty} \sum_{
\substack{i_1, i_2, \dots, i_m \in \mathbb{Z}_{\geqslant 0},\\
\alpha_1, \alpha_2, \dots, \alpha_m \in \mathbb{Z}_{2}}} 
B_{i_1, i_2, \dots, i_m}^{(\alpha_1, \alpha_2, \dots, \alpha_m)} \, 
\wh{z}_{i_1}^{(\alpha_1)}
\wh{z}_{i_2}^{(\alpha_2)}
\dots \wh{z}_{i_m}^{(\alpha_m)}, 
\end{equation*}
where only finite number of coefficients $B_{i_1, i_2, \dots, i_m}^{(\alpha_1, \alpha_2, \dots, \alpha_m)}$ 
is not zero, and $\wh{z}_{i}^{(\alpha)} = \wh{x}_{i}$, if $\alpha = \bar 0$, and 
$\wh{z}_{i}^{(\alpha)} = \wh{p}_i$, if $\alpha = \bar 1$, 
$i \in \mathbb{Z}_{\geqslant 0}$, $\mathbb{Z}_{2} = \lbrace \bar 0, \bar 1 \rbrace$. 
The commutation relations between $\wh{x}_i$ and $\wh{p}_j$, $i, j \in \mathbb{Z}_{\geqslant 0}$, 
imply that one can assume without loss of generality a certain symmetry or antisymmetry of the coefficients 
$B_{i_1, i_2, \dots, i_m}^{(\alpha_1, \alpha_2, \dots, \alpha_m)}$ with respect to the permutations of indices. 
This leads to the following construction. 

In the definitions of the $q$-analogues of the Weyl quantization, 
of the star product, and of the Poisson bracket, 
one can restrict oneself to some subspaces of 
$\mathcal{A}_{2 d}^{q} (\hbar)$ and $\mathcal{E}_{2 d}^{q} (\hbar)$ 
described, for example, in terms of $q$-symmetrizators and antisymmetrizators. 
In other words, fix a projector $\mathcal{P}: \mathcal{A}_{2 d}^{q} (\hbar) \to \mathcal{A}_{2 d}^{q} (\hbar)$, 
$\mathcal{P}^2 = \mathcal{P}$, and consider a product 
$f \circledast_{\mathcal{P}} g = \mathcal{P} (f \circledast g)$, for $f, g \in \mathcal{P} \mathcal{A}_{2 d}^{q} (\hbar)$. 
This modified product does \emph{not} need to be associative, and therefore 
the corresponding first ``semiclassical'' correction 
$\langle f, g \rangle_{\mathcal{P}}^{(q)} = \mathcal{P} (\langle f, g \rangle^{(q)})$ 
does not need to be a 2-cocycle. 
Let us write $T \langle f_1, f_2, \dots, f_n \rangle_{\mathcal{P}}^{(q)}$ 
for a bracketing of $(f_1, f_2, \dots, f_n) \in (\mathcal{A}_{2 d}^{q} (\hbar))^{n}$ 
described by a planar binary tree $T \in \mathcal{Y}$ with $n$ leaves, $|T| = n$. 
Basically, one obtains a collection of brackets 
$\lbrace T \langle -, -, \dots, - \rangle_{\mathcal{P}}^{(q)} \rbrace_{T \in \mathcal{Y}}$, 
related to each other in a more general way than the $q$-Jacobi and the $q$-Leibniz identities. 
It might be of interest to describe this collection in terms of the homotopy theory. 

One can consider the construction mentioned not on the whole classical
shadow $\mathcal{A}_{2 d}^{q} (\hbar)$, but taking a truncation
$(\mathcal{A}_{2 d}^{q} (\hbar))_{\leqslant N}$ at some level $N \in\mathbb{Z}_{> 0}$.  
As an example to have in mind, take $N = 2$ (we denote the generators as $\hbar_i$ and $\hbar_{i, j}$ in this case, 
$i, j \in [2 d]$, $i < j$), and
consider a linear map $\mathcal{P}: 
(\mathcal{A}_{2 d}^{q} (\hbar))_{\leqslant 2} \to (\mathcal{A}_{2 d}^{q} (\hbar))_{\leqslant 2}$, 
\begin{multline*} 
\mathcal{P}: \hbar_{i_1, j_1} \dots \hbar_{i_m, j_m} \hbar_{k_1} \dots \hbar_{k_n}
\mapsto \sum_{\sigma \in S_m} c_{j_1, \dots, j_m}^{i_1, \dots, i_m} (\sigma) 
\times \\ \times
\hbar_{i_1, j_{\sigma (1)}} \dots \hbar_{i_m, j_{\sigma (m)}} \hbar_{k_1} \dots \hbar_{k_n}, 
\end{multline*}
where $i_{\mu}, j_{\mu}, k_{\nu} \in [2 d]$, $\mu \in [m]$, $\nu \in [n]$, $m, n \in \mathbb{Z}_{\geqslant 0}$, 
and the coefficients $c_{j_1, \dots, j_m}^{i_1, \dots, i_m} (\sigma) = 
b_{j_1, \dots, j_m}^{i_1, \dots, i_m} (\sigma)/ 
\sum_{\rho \in S_m} (b_{j_1, \dots, j_m}^{i_1, \dots, i_m} (\sigma))^2$, where 
$b_{j_1, \dots, j_m}^{i_1, \dots, i_m} (\sigma)$ is the ``braiding'' factor 
determined by 
\begin{equation*}
\xi_{i_1} \xi_{j_{\sigma (1)}} \dots \xi_{i_m} \xi_{j_{\sigma (m)}} = 
b_{j_1, \dots, j_m}^{i_1, \dots, i_m} (\sigma) 
\xi_{i_1} \xi_{j_1} \dots \xi_{i_1} \xi_{j_m},
\end{equation*} 
for a collection symbols satisfying $\xi_i \xi_j = - q_{j, i} \xi_j \xi_i$, like $\hbar_{i, j} = - q_{j, i} \hbar_{j, i}$.

Another way to distort the classical Poisson bracket $\lbrace -, - \rbrace$ would be to consider 
more complicated quadratic commutation relations for 
the noncommutative Planck constants 
$\wh{\hbar}_{\Gamma}$, $\Gamma \in \mathcal{Y} ([2 d])$, such as the reflection equation algebra relations 
\cite{Gurevich,Gurevich_Pyatov_Saponov}. 
One can also think of more general star products in analogy with the twisted products induced by a 
coquasitriangular structure on a Hopf algebra 
\cite{Hirshfeld_Henselder,Lopez_Panaite_VanOystaeyen,Panaite_VanOystaeyen}. 
For a collection of leaf-labelled trees $\wt{\Gamma} = (\Gamma_1, \dots \Gamma_n) \in 
(\mathcal{Y}^{+} ([2 d]))^{n}$, and $I \subset [n]$, $I = \lbrace i_1 < i_2 < \dots < i_r \rbrace$, 
write $\hbar_{\wt{\Gamma}_{I}} := \hbar_{\Gamma_{i_1}} \dots \hbar_{\Gamma_{i_r}}$, and set $\bar I := [n] \backslash I$. 
Consider a bilinear product $\circledast_{R}$ on $\mathcal{A}_{2 d}^{q} (\hbar)$, 
\begin{equation*} 
\hbar_{\wt{\Gamma}_{[n]}} \circledast_{R} \hbar_{\wt{\Gamma}_{[m]}'} := 
\sum_{I \subset [n]} \sum_{J \subset [m]} 
R (\hbar_{\wt{\Gamma}_{I}}, \hbar_{\wt{\Gamma}_{J}'}) 
\hbar_{\wt{\Gamma}_{\bar I}} \hbar_{\wt{\Gamma}_{\bar J}'}, 
\end{equation*}
where $\wt{\Gamma} \in (\mathcal{Y}^{+} ([2 d]))^{n}$, 
$\wt{\Gamma}' \in (\mathcal{Y}^{+} ([2 d]))^{m}$, and 
$R$ is a bilinear map $R: \mathcal{A}_{2 d}^{q} (\hbar) \otimes \mathcal{A}_{2 d}^{q} (\hbar) \to \mathcal{A}_{2 d}^{q} (\hbar)$. 
If $\circledast_{R}$ is associative, then $R$ is determined by its restriction 
$\bar R: (\hbar_{\Gamma}, \hbar_{\Gamma'}) \mapsto R (\hbar_{\Gamma}, \hbar_{\Gamma'})$ to $V \times V$, 
where $V$ is the vector space spanned over the generators $\hbar_{\Gamma}$, $\Gamma \in \mathcal{Y}^{+} ([2 d])$. 
For the product $\star$ corresponding to the normal $q$-quantization, this is 
a map which factors through $V$, $\bar R: V \times V \to V \subset \mathcal{A}_{2 d}^{q} (\hbar)$, 
and it is a graded map of degree $+1$ 
(the grading on $V$ is given by $\mathrm{deg} (\hbar_{\Gamma}) = |\Gamma| - 1$). 
Extending the analogy with the $q$-Poisson bracket, one may say that this 
$\bar R$ corresponds to a \emph{``coquasitriangular structure multiplied by $\hbar$''}. 
Denote $\mathbb{L}_{\Gamma} (R)$  
the operators of left multiplication on the classical shadow by $\hbar_{\Gamma}$, 
$\mathbb{L}_{\Gamma} (R) = \hbar_{\Gamma} \circledast_{R} -$, $\Gamma \in \mathcal{Y} ([2 d])$. 
The commutation relations 
$\wh{\hbar}_{\Gamma} \wh{\hbar}_{\Gamma'} - 
q_{\Gamma', \Gamma} \wh{\hbar}_{\Gamma'} \wh{\hbar}_{\Gamma} = \wh{\hbar}_{\Gamma \vee \Gamma'}$  
in $\mathcal{E}_{2 d}^{q} (\hbar)$, 
where $\Gamma, \Gamma' \in \mathcal{Y}^{+} ([2 d])$, 
are translated into a system 
\begin{equation*} 
\mathbb{L}_{\Gamma} (R) \mathbb{L}_{\Gamma'} (R) - 
q_{\Gamma', \Gamma} \mathbb{L}_{\Gamma'} (R) \mathbb{L}_{\Gamma} (R) = 
\mathbb{L}_{\Gamma \vee \Gamma'} (R), 
\end{equation*}
where $\Gamma$ and $\Gamma'$ vary over $\mathcal{Y}^{+} ([2 d])$, 
which can be perceived as a generalized quantum Yang-Baxter equation for 
$\mathcal{E}_{2 d}^{q} (\hbar)$.

\section{Discussion}

As has already been pointed out in the introduction, 
the problem of quantization on a noncommutative space, 
as well as the concept of the/a noncommutative geometry 
\cite{Connes,Drinfeld,Kontsevich_1,Kontsevich_2,LeBruyn,VanOystaeyen_Verschoren}
itself, admits various approaches. 
In the most simple case, the philosophy described in the present paper can be expressed as follows. 
If one starts with a Poisson bracket $\lbrace -, - \rbrace$, say, on an algebra 
$\mathbb{C} [x_1, \dots, x_d, p_1, \dots, p_d]$ of polynomial observables of a classical mechanical system, 
then what one should do first of all is to span a vector space $\bar V$ over a collection of symbols 
$\bar x_1, \bar x_2, \dots, \bar x_d$ and $\bar p_1, \bar p_2, \dots, \bar p_d$, 
and to embed it into a $\mathbb{Z}$-graded vector space $V$ as 
a component of degree zero, $\bar V \subset V$. 
It is suggested to construct the ambient  
vector space $V$ as a span over the 
symbols $\hbar_{\Gamma}$ indexed by the planar binary leaf-labelled trees $\Gamma \in \mathcal{Y}^{+} ([2 d])$ 
(the notation is described in the main text), $\mathrm{deg} (\hbar_{\Gamma}) = |\Gamma| - 1$. 
In particular, $\bar x_i$ and $\bar p_j$, $i, j \in [d]$, correspond to $\hbar_{\Gamma}$ with $|\Gamma| = 1$.   
Next, one defines a filtered algebra $\mathcal{E}_{2 d} (\hbar)$ generated by the symbols $\wh{\hbar}_{\Gamma}$, 
$\Gamma \in \mathcal{Y}^{+} ([2 d])$ (the \emph{noncommutative Planck constants}) satisfying the relations 
$[\wh{\hbar}_{\Gamma}, \wh{\hbar}_{\Gamma'}] = \wh{\hbar}_{\Gamma \vee \Gamma'}$, for 
$\Gamma, \Gamma' \in \mathcal{Y}^{+} ([2 d])$, where the filtration is the decreasing filtration 
induced by the length of polynomials filtration on the tensor algebra $T (V)$. 
A \emph{quantization} is a linear map 
\begin{equation*} 
W: \mathcal{A}_{2 d} (\hbar) \to \mathcal{E}_{2 d} (\hbar), 
\end{equation*} 
where $\mathcal{A}_{2 d} (\hbar)$ is the associated graded of $\mathcal{E}_{2 d} (\hbar)$ 
(the \emph{classical shadow}), 
satisfying the left and right 
affine 
coequivariance and the existence of the classical limit conditions. 
Conceptually, an analogue of a Poisson bracket is the \emph{first semiclassical correction} to the 
star product $\circledast$ on the classical shadow, $W (f \circledast g) = W (f) W (g)$, 
described by 
\begin{equation*} 
\langle -, - \rangle: \mathcal{A}_{2 d} (\hbar) \otimes \mathcal{A}_{2 d} (\hbar) \to \mathcal{A}_{2 d} (\hbar), 
\end{equation*}
which is a graded linear map of degree $+1$ with respect to the grading induced by the $\mathbb{Z}$-grading on $V$.   

Intuitively, the algebra $\mathcal{E}_{2 d} (\hbar)$ is a kind of \emph{noncommutative neighbourhood} of 
the semiclassical parameter $\hbar$, 
\begin{equation*} 
\hbar \quad \longrightarrow \quad 
\wh{\hbar}_{i}, \quad \wh{\hbar}_{[i, j]}, \quad 
\wh{\hbar}_{[[i, j], k]}, \quad 
\wh{\hbar}_{[[i, j], [k, l]]}, \quad \dots, 
\end{equation*}
where $i, j, k, l, \dots \in [2 d]$, corresponding to an idea 
to replace every next commutator with a new variable. 
In fact, the nature of the parameter $\hbar$ is not important, and one can consider the same construction 
around $\lambda^{-1}$, where $\lambda \to \infty$ is the rescaling parameter for renormalization. 
The advantage of our approach is that it admits a natural Hopf algebraic deformation 
(in this paper we describe the $q$-deformation $\mathcal{E}_{2 d}^{q} (\hbar)$, where $q$ is a 
$2 d \times 2 d$ matrix of formal variables). 
It would be of interest to study this in more detail, but we leave it for another paper.


{\tiny This work is supported by 
FWO (Fonds Wetenschappelijk Onderzoek -- Vlaanderen), 
the research project G.0622.06 
(``Deformation quantization methods for algebras and categories with applications to quantum mechanics'').
}


\begin{thebibliography}{99}


\bibitem{Alekseev_KosmannSwarzbach}
Alekseev, A.; Kosmann-Schwarzbach, Y.: 
{Manin pairs and moment maps}.
J. Differential Geom.  \textbf{56}, no. 1, 133--165 (2000) 


\bibitem{Alekseev_KosmannSwarzbach_Meinrenken}
Alekseev, A.; Kosmann-Schwarzbach, Y.; Meinrenken, E.: 
{Quasi-Poisson manifolds}.  
Canad. J. Math.  \textbf{54}, no. 1, 3--29 (2002)


\bibitem{Beggs_Majid}
Beggs, E. J. and Majid, S.: 
{Semiclassical differential structures}.   
Pacific J. Math.  \textbf{224}, no. 1, 1--44 (2006)


\bibitem{VandenBergh}
Van den Bergh, M.:  
{Double Poisson algebras}.   
Trans. Amer. Math. Soc.  \textbf{360} no. 11, 5711--5769 (2008) 



\bibitem{Bocklandt_LeBruyn}
Bocklandt, R. and  Le Bruyn, L.:  
{Necklace Lie algebras and noncommutative symplectic geometry}.   
Math. Z.  \textbf{240},  no. 1, 141--167 (2002) 


\bibitem{Castro}
Castro, C.: 
{On modified Weyl-Heisenberg algebras, noncommutativity, matrix-valued Planck constant and QM in Clifford spaces}.  
J. Phys. A  \textbf{39},  no. 45, 14205--14229 (2006) 


\bibitem{Chaichian_Demichev_Kulish}
Chaichian, M., Demichev, A. P., and Kulish, P. P.:  
{Quasi-classical limit in $q$-deformed systems, non-commutativity and the $q$-path integral}.   
Phys. Lett. A  \textbf{233},  no. 4-6, 251--260   (1997) 

\bibitem{Connes}
Connes, A.: 
{\it Cyclic cohomology and noncommutative differential geometry}.   
Proceedings of the International Congress of Mathematicians, 
Vol. \textbf{1} Berkeley, Calif., 1986,  pp. 879--889   


\bibitem{CrawleyBoevey_Etingof_Ginzburg}
Crawley-Boevey, W., Etingof, P., and Ginzburg, V.:  
{Noncommutative geometry and quiver algebras}.   
Adv. Math.  \textbf{209},  no. 1, 274--336  (2007) 


\bibitem{Drinfeld} 
Drinfel'd, V. G.: 
{Quasi-Hopf algebras}.  
(Russian)  
Algebra i Analiz  \textbf{1},  no. 6, 114--148   (1989),   
translation in 
Leningrad Math. J.  \textbf{1},  no. 6, 1419--1457   (1990)



\bibitem{Farkas}
Farkas, D. R.: 
{A ring-theorist's description of Fedosov quantization}.  
Lett. Math. Phys. \textbf{51}, no. 3, 161--177 (2000) 


\bibitem{Fialowski_Penkava} 
Fialowski, A. and Penkava, M.:  
{Deformation theory of infinity algebras}.   
J. Algebra  \textbf{255},  no. 1, 59--88 (2002) 


\bibitem{Ginzburg}
Ginzburg, V.: 
{Non-commutative symplectic geometry, quiver varieties, and operads}.  
Math. Res. Lett. \textbf{8}, no. 3, 377--400 (2001)

\bibitem{Gurevich}
Gurevich, D. I.: 
{Algebraic aspects of the quantum Yang-Baxter equation}.  
(Russian) 
Algebra i Analiz  \textbf{2},  no. 4, 119--148 (1990),   
translation in  
Leningrad Math. J.  \textbf{2},  no. 4, 801--828 (1991) 

\bibitem{Gurevich_Pyatov_Saponov}
Gurevich, D. I.,  Pyatov, P. N., and Saponov, P. A.:  
{Representation theory of an algebra of a (modified) reflection equation of $\mathrm{GL}(m\vert n)$ type}.  
(Russian)  
Algebra i Analiz  \textbf{20},  no. 2, 70--133 (2008),  
translation in 
St. Petersburg Math. J.  \textbf{20},  no. 2, 213--253 (2009)



\bibitem{Hartwig_Larsson_Silvestrov}
Hartwig, J. T., Larsson, D., and Silvestrov, S. D.:  
{Deformations of Lie algebras using $\sigma$-derivations}. 
J. Algebra  \textbf{295},  no. 2, 314--361 (2006) 


\bibitem{Hirshfeld_Henselder}
Hirshfeld, A. C. and Henselder, P.:  
{Star products and quantum groups in quantum mechanics and field theory}.   
Ann. Physics  \textbf{308},  no. 1, 311--328 (2003)


\bibitem{Kapranov}
Kapranov, M.: 
{Noncommutative geometry based on commutator expansions}.  
J. Reine Angew. Math.  \textbf{505}, 73--118, (1998)

\bibitem{Maslov_Karasev}
Karasev, M. V.; Maslov, V. P.: 
{Nonlinear Poisson brackets. Geometry and quantization}. 
Translated from the Russian by A. Sossinsky and M. Shishkova. 
Translations of Mathematical Monographs, 119. American Mathematical Society, Providence, RI, xii+366 pp., 1993. 


\bibitem{Kontsevich_1}
Kontsevich, M.: 
{Deformation quantization of Poisson manifolds, I}.  
\texttt{arXiv:q-alg/9709040} 

\bibitem{Kontsevich_2}
Kontsevich, M.: 
{\it Formal (non)commutative symplectic geometry}.  
The Gel'fand Mathematical Seminars, 1990--1992,  
Birkhäuser Boston, Boston, MA, 1993, pp. 173--187 


\bibitem{Kubo}
Kubo, F.:  
{Finite-dimensional non-commutative Poisson algebras}.   
J. Pure Appl. Algebra  \textbf{113},  no. 3, 307--314  (1996)

\bibitem{LeBruyn}
Le Bruyn, L.: 
{Noncommutative geometry and Cayley-smooth orders}. 
Pure and Applied Mathematics (Boca Raton), 290. 
Chapman \& Hall/CRC, Boca Raton, FL, 2008 


\bibitem{Lopez_Panaite_VanOystaeyen}
L\'opez Pe\~na, J., Panaite, F., and Van Oystaeyen, F.:  
{General twisting of algebras}.  
Adv. Math.  \textbf{212},  no. 1, 315--337 (2007) 



\bibitem{Maslov_hyper}
Maslov, V. P.: 
{An application of the method of ordered operators to the problem of obtaining exact solutions}. 
(Russian)  
Teoret. Mat. Fiz.  \textbf{33}, no. 2, 185--209 (1977)

\bibitem{Maslov_opmeth}
Maslov, V. P.:
{Operational methods}.
Translated from the Russian by V. Golo, N. Kulman and G. Voropaeva. 
Mir Publishers, Moscow,  559 pp., 1976 


\bibitem{Maslov_char}
Maslov, V. P.: 
{Nonstandard characteristics in asymptotic problems}. 
(Russian)  
Uspekhi Mat. Nauk  \textbf{38},  no. 6 (234), 3--36  (1983)



\bibitem{Newmann}
Newmann, M. H. A.: 
{On theories with a combinatorial definition of ``equivalence''}.
Ann. of Math. \textbf{43}, 223--243 (1942) 



\bibitem{VanOystaeyen_Verschoren}
Van Oystaeyen, F. and Verschoren, A.: 
{Noncommutative algebraic geometry. An introduction}. 
Lecture Notes in Mathematics, 887. Springer-Verlag, Berlin, 1981  



\bibitem{Panaite_VanOystaeyen}
Panaite, F. and Van Oystaeyen, F.:  
{Quasi-Hopf algebras and representations of octonions and other quasialgebras}.   
J. Math. Phys.  \textbf{45},  no. 10, 3912--3929 (2004)


\bibitem{Penkava}
Penkava, M.: 
{L-infinity algebras and their cohomology}. 
\texttt{arXiv:q-alg/ 9512014v1} 


\bibitem{Pichereau_VandeWeyer}
Pichereau, A. and Van de Weyer, G.:  
{Double Poisson cohomology of path algebras of quivers}.   
J. Algebra  \textbf{319},  no. 5, 2166--2208  (2008)


\bibitem{Richard_Silvestrov}
Richard, L. and Silvestrov, S. D.:  
{Quasi-Lie structure of $\sigma$-derivations of $\mathbb{C}[t^{\pm 1}]$}.   
J. Algebra  \textbf{319},  no. 3, 1285--1304 (2008) 





\bibitem{Xu}
Xu, P.: 
{Noncommutative Poisson algebras}.  
Amer. J. Math.  \textbf{116},  no. 1, 101--125 (1994)

\bibitem{Yao_Ye_Zhang}
Yao, Y., Ye, Y, and Zhang, P.:  
{Quiver Poisson algebras}.  
J. Algebra  \textbf{312},  no. 2, 570--589 (2007)






\end{thebibliography}
\end{document}